\documentclass[a4paper,fleqn]{cas-sc}

\usepackage[authoryear,longnamesfirst]{natbib}

\usepackage{amsthm}

\usepackage[authoryear,longnamesfirst]{natbib}

\newcounter{casecount}
\newproof{pf}{proof}
\newtheorem{theorem}{Theorem}[section]
\newtheorem{lemma}[theorem]{Lemma}

\DeclareMathOperator*{\starspan}{span}
\newtheorem{problem}{Problem}
\newtheorem{remark}[casecount]{Remark}

\usepackage{etoolbox}
\usepackage{color}

\numberwithin{equation}{section}

\usepackage{tikz}
\usepackage{pgfplots}
\pgfplotsset{compat=newest}
\pgfplotsset{plot coordinates/math parser=false}
\newlength\figureheight
\newlength\figurewidth
\usepackage{stackengine}
\usepackage{multirow}
\usepackage{graphicx}
\usepackage{float}
\usepackage{subfig}
\graphicspath{{./results/}}
\usepackage[ruled,vlined,linesnumbered]{algorithm2e}
\usepackage{algorithmic}

\newcommand{\redline}{\raisebox{2pt}{\tikz{\draw[-,red,line width=1.2pt](0,0)--(6mm,0);}}}
\newcommand{\blueline}{\raisebox{2pt}{\tikz{\draw[-,blue,dotted,line width=1.2pt](0,0)--(6mm,0);}}}
\newcommand{\greenline}{\raisebox{2pt}{\tikz{\draw[-,green,dashed,line width=1.2pt](0,0)--(6mm,0);}}}

\newcommand{\magentaline}{\raisebox{2pt}{\tikz{\draw[-,magenta,dash dot,line width=1.2pt](0,0)--(6mm,0);}}}

\begin{document}
\let\WriteBookmarks\relax
\def\floatpagepagefraction{1}
\def\textpagefraction{.001}
\shorttitle{Implicit Higher-Order Moment Matching }
\shortauthors{M.M.A. Asif et~al.}

\title [mode = title]{Implicit Higher-Order Moment Matching Technique for Model Reduction of Quadratic-bilinear Systems}                      

\tnotetext[cor2]{\ This author is supported by HEC Pakistan under NRPU Project ID 10176.}
\tnotetext[cor1]{\ This author is supported by HEC Pakistan under NRPU Project ID 10176. Also a research visit to Magdeburg was sponsored by DAAD.}
\tnotetext[cor3]{\ This author is supported by DFG Research Training Group RTG 2297/1 ''Mathematical Complexity Reduction- MathCoRe''.}
\author[1,2]{Mian Muhammad Arsalan Asif}[type=editor]
\cormark[1]
\address[1]{School of Electrical Engineering and Computer Sciences, National University of Sciences and Technology, Islamabad, Pakistan}

\author[2]{Mian Ilyas Ahmad}
\address[2]{Research Center of Modelling and Simulation, National University of Sciences and Technology, Islamabad, Pakistan}
\ead[2]{m.ilyas@rcms.nust.edu.pk}
\cormark[2]

\author[3]{Peter Benner}
\address[3]{Max Planck Institute for Dynamics of Complex Technical Systems, Magdeburg,
Germany, and Faculty of Mathematics, Otto-von-Guericke University Magdeburg, Germany}
\cormark[3]
\author[3]{Lihong Feng}
\author[4]{Tatjana Stykel}
\address[4]{Institut f\"ur Mathematik, Unversit\"at Augsburg, Universit\"atsstr. 14, 86159 Augsburg, Germany}
\nonumnote{Corresponding author}

\begin{abstract}We propose a projection based multi-moment matching method for model order reduction of quadratic-bilinear systems. The goal is to construct a reduced system that ensures higher-order moment matching for the multivariate transfer functions appearing in the input-output representation of the nonlinear system. An existing technique achieves this for the first two multivariate transfer functions, in what is called the symmetric form of the multivariate transfer functions. We extend this framework to an equivalent and simplified form, the regular form, which allows us to show moment matching for the first three multivariate transfer functions. Numerical results for three benchmark examples of quadratic-bilinear systems show that the proposed framework exhibits better performance with reduced computational cost in comparison to existing techniques.
\end{abstract}
%
%
\maketitle
\section{Introduction}
 We consider the problem of model order reduction for a single-input single-output (SISO) quadratic-bilinear descriptor system of the form
\begin{equation}
\begin{aligned}\label{fom_qbdae}
E\dot{x}(t)&=Ax(t)+Nx(t)u(t)+H\bigl(x(t)\otimes x(t)\bigr)+Bu(t),\\
y(t)&=Cx(t), \qquad x(0)=0
\end{aligned}
\end{equation}
where $E,A,N\in \mathbb{R}^{n\times n}$, $H\in \mathbb{R}^{n\times n^{2}}$, $B\in \mathbb{R}^{n\times 1}$,  $C\in \mathbb{R}^{1\times n}$, $x(t)\in \mathbb{R}^{n}$ is the state vector, and $u(t), y(t)\in \mathbb{R}$ 
are the input and the output, respectively.
 If the quadratic matrix $H$ and the bilinear matrix $N$ are omitted, the system will reduce to a linear 
 state-space system. 

There are many applications where the system can be represented by quadratic-bilinear models. These include flow problems in energy networks, VLSI circuit design, interaction of biological systems and chemical processes.
Also a large class of nonlinear systems including fractions, exponentials, logarithmic and power terms can be represented in the quadratic-bilinear form by using exact transformations \cite{morGu11}, which increases their application range even further. Often, the requirement of these applications is to construct large scale quadratic-bilinear models, which are computationally expensive to analyse, control or optimise. Model order reduction provides a remedy to this problem.

Model order reduction is a process to compute another quadratic-bilinear descriptor system of the form
\begin{equation}
\begin{aligned}\label{ros}
\hat{E}\dot{\hat{x}}(t)&=\hat{A}\hat{x}(t)+\hat{N}\hat{x}(t)u(t)+
\hat{H}\bigl(\hat{x}(t)\otimes \hat{x}(t)\bigr)+\hat{B}u(t),\\
\hat{y}(t)&=\hat{C}\hat{x}(t), \qquad x(0)=0,
\end{aligned} 
\end{equation}
which has the reduced order $\hat{n}\ll n$, and $\hat{y}(t)\approx y(t)$ for all admissible inputs $u(t)$.  Here, the reduced state vector $\hat{x}(t)\in \mathbb{R}^{\hat{n}}$ and the reduced system matrices are $\hat{E},\hat{A},\hat{N}\in \mathbb{R}^{\hat{n}\times \hat{n}}$,  $\hat{H}\in \mathbb{R}^{\hat{n}\times \hat{n}^{2}}$, $\hat{B}\in \mathbb{R}^{\hat{n}\times 1}$ and $\hat{C}\in \mathbb{R}^{1\times \hat{n}}$.

Projection is commonly used to construct the reduced-order system, where two matrices $V$ and $W$ are identified such that their columns span 
$\hat{n}$-di\-men\-sio\-nal subspaces $\mathcal{V}$ and $\mathcal{W}$, respectively. In particular, projection involves the following steps:
\begin{itemize}
\item approximate the state vector in $\mathcal{V}$ such that $x(t) \approx V \hat{x} (t)$;
\item ensure Petrov-Galerkin condition so that the residual 
$$
\qquad\quad r(t) =EV \dot{\hat{x}}(t)-\Big(AV \hat{x}(t) + NV \hat{x}(t)u(t) + H\bigl(V \hat{x}(t) \otimes V \hat{x}(t)\bigr) + Bu(t)\Big), 
$$
is orthogonal to $\mathcal{W}$, i.e., $W^Tr(t)=0$.
\end{itemize}
This leads to the reduced system matrices
\begin{equation}\label{red_eq}
\begin{aligned}
 \hat{E}&=W^{T}EV, & \quad \hat{A}&=W^{T}AV, & \quad\hat{H}&=W^{T}H(V\otimes V),\\
\hat{N}&=W^{T}NV, &  \quad \hat{B}&=W^{T}B,  & \quad\hat{C}&=CV.
\end{aligned}
\end{equation}
Clearly, the choice of $V$ and $W$ dictates the accuracy of the reduced system for a~given quadratic-bilinear system. In case of linear systems, the choice of $V$ and $W$ is linked to the transfer function of the system. That is, $V$ and $W$ can be selected such that some of the moments associated with the transfer function of the original and reduced systems are matched, c.f. \cite{morGri97}. However, in the quadratic-bilinear case, the $s$-domain representation involves a series of multivariate transfer functions. Here, the series grows with a new $s$-variable and the complete representation involves multiple $s$-variables, $s_1,s_2,\ldots$. The structure of these multivariate transfer functions becomes increasingly complex as the series grows. For simplification, often the concept of moment matching is restricted to the first two multivariate transfer functions. 

An approach using one-sided projection (i.e. $V=W$) with moment matching of the first two transfer functions of a~quadratic-bilinear system was discussed initially in \cite{morGu11}. An~extension to two-sided projection has been presented in \cite{morBenB15} that refines the quality of the approximations in term of accuracy. It is observed in  \cite{morAhmBJ16} that a~simplified and equivalent representation of the multivariate transfer functions can be identified, for which it is relatively easy to extend the moment matching concept to higher multivariate transfer functions. In all these approaches, the first two or higher multivariate transfer functions and their first derivatives are matched by the reduced-order system. Recently, a new approach has been proposed  in \cite{morYanJ18} that matches not only the first two moments but also higher-order moments of the first two multivariate transfer functions. Although the approach is similar to a simplified version of \cite{morBenB12a}, the method in \cite{morYanJ18} shows higher-order moment matching properties without the use of the matrices $N$ and $H$ in the construction of the projection matrices
$V$ and $W$. Recently in \cite{morBenCS19, xcao}, it has been observed that structured non-linear systems (bilinear and quadratic-bilinear) can be linked to parametric LTI systems and solved through optimisation methods. This new framework is interesting but the computational cost is often high. 

In this paper, we propose a new reduction technique that matches higher-order moments and utilises the simplified structure of the multivariate transfer functions proposed in \cite{morAhmBJ16}. The simplified structure allows us to extend the multi-moment matching concept to the first three transfer functions. Also the construction of the projection matrices $V$ and $W$ now requires the use of the matrices $N$ and $H$ 
from the nonlinear terms, and will result in more accurate reduced-order system. This is different from \cite{morBenB12a}, where multi-moment matching of the first two transfer functions is discussed. The approximation quality of the reduced-order system produced by the proposed method is compared to those obtained by existing moment matching approaches from \cite{morBenB15, morAhmBJ16, morYanJ18}. It is observed that the proposed approach provides reduced-order models with smaller approximation error.

The remaining part of the paper is organized as follows. In Section~\ref{sec:background}, we discuss the concept of multivariate moment matching for quadratic-bilinear systems by introducing different forms of multi-variate transfer functions and presenting some of the existing techniques from the literature for moment matching. In Section~\ref{sec:matching2}, the concept of generalized multi-moment matching has been utilized for regular transfer functions. Its extension to third regular transfer function will be discussed in Section~\ref{sec:matching3}. Finally, in Section~\ref{sec:comparison}, we draw a comparison with existing techniques using some benchmark examples and conclude our findings. 

\section{Background}
\label{sec:background}
In this section, we present different forms of the multivariate transfer functions, their relations and some of the existing moment matching techniques for model order reduction of quadratic-bilinear systems.

\subsection{Multivariate Frequency Representation} 

The input-output representation for the SISO quadratic-bilinear system \eqref{fom_qbdae} can be expressed by the Volterra series expansion of the output $y(t)$ with quantities analogous to the standard convolution operator
\begin{equation} \label{eq_tf_basic}
y(t)=\sum_{k=1}^{\infty}\int_{0}^{t}\int_{0}^{t_1}\int_{0}^{t_{k-1}}h_{k}(t_{1},...,t_{k})u(t-t_{1})...u(t-t_{k})dt_{k}...dt_{1},
\end{equation}
where the generalised impulse response $h_k$, also referred to as the $k$-dimensional kernel of the $k$-th subsystem, and the input $u$ are assumed to be one-sided, i.e., $h_{k}(t_{1},...,t_{k})=0$ for $t_{i}<0$, $i=0,...,k$, and  $u(t)=0$ for $t<0$. Applying the convolution property of the multi-variable Laplace transform to \eqref{eq_tf_basic}, we obtain
\begin{equation}\label{tf1}
Y_k(s_{1},...,s_{k})=H_{k}(s_{1},...,s_{k})U(s_{1})\cdots U(s_{k}),
\end{equation}
where $H_{k}(s_{1},...,s_{k})$ is the multivariate transfer function of the $k$-dimensional subsystem, see \cite{nonlinearRugh81} for details. Given the Laplace transform of the input $U(s)$ and the multivariate transfer function $H_{k}(s_{1},...,s_{k})$, the output of the $k$-th subsystem $y_{k}(t_{1},...,t_{k})$ can be identified through the inverse Laplace transform and the output 
of the system becomes 
\begin{equation}\label{tri_tf}
y(t)= \left. \sum_{k=1}^{\infty}y_{k}(t_{1},...,t_{k}) \right\vert_{t_{1}=\cdots=t_{k}=t}=\sum_{k=1}^{\infty}y_{k}(t,...,t).
\end{equation}
Note that the generalized impulse response can be written in different forms, for example by using change of variables, and therefore the multivariate transfer function has different forms. Three special forms are the symmetric, triangular and regular form of the multivariate transfer functions as examined in \cite{nonlinearRugh81}.
The relationship between the triangular form and the symmetric form of the $k$-dimensional transfer function can be written as
\begin{equation} \label{sym2tri}
H_{ksym}(s_{1},...,s_{k})=\frac{1}{n!}\sum_{\pi(\cdot)}H_{ktri}(s_{\pi(1)},...,s_{\pi(k)}),
\end{equation}
where the summation of $\pi(\cdot)$ denotes all $k!$ permutations of $s_{1},...,s_{k}$. 
The structure of the multivariate symmetric transfer functions can be identified by using the growing exponentials approach \cite{nonlinearRugh81}. To represent the structure of these transfer functions in compact form, we define the following matrix-valued functions
	\begin{align*}
X_{j}(s)=[(sE-A)^{-1}E]^{j}(sE-A)^{-1}, \qquad j= 0,1,2,\ldots\, .
\end{align*}
Their derivatives are given by 
\begin{equation} \label{lemma_eq1}
\frac{\rm d^\ell}{{\rm d} s^\ell}X_j(s)=(-1)^\ell\frac{(j+\ell)!}{j!}X_{j+\ell}(s), \qquad \ell=0,1,2,\ldots\, .
\end{equation}
With these notations, the first two subsystems of the quadratic-bilinear system \eqref{fom_qbdae} can be written as
\begin{align} \label{hsym}
\begin{split}
H_{1sym}(s_{1}) &=C(s_{1}E-A)^{-1}B=CX_{0}(s_{1})B,\\
H_{2sym}(s_{1},s_{2}) &=\frac{1}{2!} CX_{0}(s_{1}+s_{2})\Bigl(N\bigl(X_{0}(s_{1})B+X_{0}(s_{2})B\bigr)\\
&\quad + H\bigl(X_{0}(s_{1})B\otimes X_{0}(s_{2})B+X_{0}(s_{2})B\otimes X_{0}(s_{1})B\bigr)\Bigr).
\end{split}
\end{align}
The multivariate transfer function of the third subsystem will be discussed later. For higher subsystems, the structure involves further combinations of the multivariate functions, making the transfer function increasingly complex. An alternate form of the multivariate transfer function is the regular form which has relatively few terms \cite{morAhmBJ16} and is well-used for bilinear systems \cite{morBreD10}. The  $k$-dimensional regular form of the transfer function can be linked to the corresponding triangular form by using
\begin{equation}\label{tri2reg}
H_{ktri}(s_{1},\ldots,s_{k})=H_{kreg}(s_{1},s_{1}+s_{2},\ldots,s_{1}+s_{2}+\cdots+s_{k}).
\end{equation}
The relationships in \eqref{sym2tri} and \eqref{tri2reg} are utilized in \cite{morAhmBJ16} to identify the structure of the regular transfer functions for quadratic-bilinear systems.  The first two multivariate regular transfer functions are
\begin{align}\label{reg_tf123}
\begin{split}
H_{1reg}(s_{1}) =&CX_{0}(s_{1})B,
\\
H_{2reg}(s_{1},s_{2}) =&CX_{0}(s_{2})\Bigl(NX_{0}(s_{1})B+H\bigl(X_{0}(s_{2}-s_{1})B\otimes X_{0}(s_{1})B\bigr)\Bigr).
\end{split}
\end{align}
Clearly, the regular form has less terms compared to the symmetric form. 
\subsection{Moment matching model reduction}
The output $y(t)$ of the quadratic-bilinear system \eqref{fom_qbdae} can be well approximated by the output $\hat{y}(t)$ of the reduced-order system \eqref{ros}, if the model order reduction approach ensures 
\begin{align}\label{morg}
H_k(s_{1},...,s_{k})\simeq \hat{H}_k(s_{1},...,s_{k}), \qquad k=1,...,K.
\end{align}
In practice, this approximation is achieved by ensuring that the multivariate transfer functions of the original and reduced system are matched at some predefined interpolation points $\sigma_{ki}\in \mathbb{C}$,  $i=1,\ldots,r$, for each $k$. The approximation in \eqref{morg} improves further if some of the partial derivatives of $H_k(s_{1},...,s_{k})$ with respect to $s_j$, $j=1, \ldots , k$, are also matched by the corresponding partial derivatives of $\hat{H}_k(s_{1},\ldots,s_{k})$ at the same interpolation points. Based on the required level of approximation, different approximation problems have been addressed in the literature.  For example, the problem of Hermite interpolation can be defined as follows.

\begin{problem} \label{prbenner}
For a~set of interpolation points $\sigma_{i}\notin \Lambda(A,E)\cup\Lambda(\hat{A},\hat{E})$, $i=1,\ldots, m$, find projection matrices ${V, W\in\mathbb{R}^{n\times \hat{n}}}$ with $\hat{n}\ll n$ such that
\begin{align*}
H_1(\sigma_{i})  &= \hat{H}_1(\sigma_{i}), \\
\left.\frac{\partial}{\partial s_1} H_1(s_{1})\right|_{s_{1}=\sigma_{i}}  & =
\left.\frac{\partial}{\partial s_1} \hat{H}_1(s_{1})\right|_{s_{1}=\sigma_{i}}, \\
H_2(\sigma_{i}, \sigma_{i}) & = \hat{H}_2(\sigma_{i},\sigma_{i}), \\
\left.\frac{\partial}{\partial s_{j}} H_2(s_{1},s_{2})\right|_{
 \stackon{\scriptsize$s_{2}=\sigma_{i}$}{\scriptsize$s_{1}=\sigma_{i}$}}&
\left.=\frac{\partial}{\partial s_{j}} \hat{H}_2(s_{1},s_{2})\right|_{
 \stackon{\scriptsize$ s_{2}=\sigma_{i}$}{\scriptsize$s_{1}=\sigma_{i}$}}, \qquad j=1,2.
\end{align*}
\end{problem}

This problem has been addressed in \cite{morBenB15,morGu11}, where the symmetric form of the multivariate transfer functions has been utilised for ensuring interpolation. Essentially, the result in \cite{morBenB15} can be summarised as follows.

 \begin{theorem}\label{thbenner}
For a~given set of interpolation points $\sigma_{i}\notin \Lambda(A,E)\cup\Lambda(\hat{A},\hat{E})$, $i=1,\ldots, m$, let $V,W \in \mathbb{R}^{n\times \hat{n}}$ satisfy
\begin{align*}
\mbox{\rm im}(V)=& \starspan_{i=1,...,m}\Bigl\{X_{0}(\sigma_{i})B, X_{0}(\sigma_{i}) \bigl[(H\bigl(X_{0}(\sigma_{i})B\otimes X_{0}(\sigma_{i})B\bigr)+NX_{0}(\sigma_{i})B\bigr]\Bigr\},\\
\mbox{\rm im}(W)=& \starspan_{i=1,...,m}\Bigl\{X_{0}(2\sigma_{i})^TC^T,X_{0}(\sigma_{i})^T\bigl[H^{(2)}\bigl(X_{0}(\sigma_{i})B\otimes X_{0}(2\sigma_{i})^TC^T\bigr)+\frac{1}{2} N^{T}X_{0}(2\sigma_{i})^TC^T\bigr]\Bigr\},
\end{align*}
where $H^{(2)}$ is the mode-$2$ matricization of a $3$-dimensional tensor $\mathcal{H}\in \mathbb{R}^{n\times n \times n}$ for which $H=H^{(1)}$ is the mode-$1$ matricization.
Then the reduced-order system satisfies
\begin{equation}\label{prrbenner}
\begin{aligned}
H_{1sym}(\sigma_{i})&=\hat{H}_{1sym}(\sigma_{i}),&
H_{2sym}(\sigma_{i},\sigma_{i})&=\hat{H}_{2sym}(\sigma_{i},\sigma_{i}),\\
\left.\frac{\partial}{\partial s_{j}} H_{2sym}(s_1,s_2)\right|_{
 \stackon{\scriptsize$s_{2}=\sigma_{i}$}{\scriptsize$s_{1}=\sigma_{i}$}}&= \left.\frac{\partial}{\partial s_{j}} \hat{H}_{2sym}(s_1,s_2)\right|_{
 \stackon{\scriptsize$s_{2}=\sigma_{i}$}{\scriptsize$s_{1}=\sigma_{i}$}}, & j&=1,2.
\end{aligned}
\end{equation}
\end{theorem}

Note that the approach in \cite{morBenB15} only partially solves Problem~\ref{prbenner} as it does not ensure that the reduced system matches the derivatives of the first transfer function. To resolve this, \cite{morAhmFB19} proposes a modified framework, where the reduced system also matches the derivative of the first transfer function. However, the issue of extending the framework to higher subsystems remains, as the structure of the symmetric transfer functions becomes increasingly complex for higher subsystems. 
To address this issue, the regular form of the multivariate transfer functions is used in \cite{morAhmBJ16} to solve Problem~\ref{prbenner}. 
Another possible extension of Problem~\ref{prbenner} is to match higher derivatives (higher-order moments) of the first two multivariate transfer functions, which is discussed in the following subsection.

\subsection{Higher-Order Moment Matching}

To achieve better reduced-order models for a~given set of interpolation points, the problem of generalised moment matching (where higher derivatives of the multivariate transfer functions are also matched) has been considered in the literature. 
Formally, this problem can be stated as follows.

\begin{problem}\label{prmulti}
For a~set of interpolation points $\sigma_{i}\notin \Lambda(A,E)\cup\Lambda(\hat{A},\hat{E})$, $i=1,\ldots, m$, 
find projection matrices ${V,W \in \mathbb{R}^{n\times \hat{n}}}$ such that 
\begin{align*}
\left.\frac{\partial^p}{\partial s_1^p} H_1(s_{1})\right|_{s_{1}=\sigma_{i}}  &=\left.\frac{\partial^p}{\partial s_1^p} \hat{H}_1(s_{1})\right|_{s_{1}=\sigma_{i}},\\
\left.\frac{\partial^{p+q}}{\partial s_{1}^{p}\partial s_{2}^{q}} H_2(s_{1},s_{2})\right|_{
 \stackon{\scriptsize$s_{2}=\sigma_{i}$}{\scriptsize$s_{1}=\sigma_{i}$}}
 &=\left.\frac{\partial^{p+q}}{\partial s_{1}^{p}\partial s_{2}^{q}} \hat{H}_2(s_{1},s_{2})\right|_{ \stackon{\scriptsize$s_{2}=\sigma_{i}$}{\scriptsize$s_{1}=\sigma_{i}$}},
\end{align*}
for  $p=0,\ldots,P$ and $q=0,\ldots,Q$. The upper limits $P$ and $Q$ represent the highest moments being matched. 
\end{problem}

It is easy to see that Problem~\ref{prmulti} reduces to Problem~\ref{prbenner} for  $P=Q=1$ and $p\neq q>0$. Problem~\ref{prmulti} has been addressed recently in \cite{morYanJ18}, where the first two symmetric transfer functions have been used to match higher derivatives. The main result of \cite{morYanJ18} is summarised in the following theorem.

\begin{theorem}\label{thtnf}
For a~given set of interpolation points $\sigma_{i}\notin \Lambda(A,E)\cup\Lambda(\hat{A},\hat{E})$, $i=1,\ldots, m$, let $V,W \in \mathbb{R}^{n\times \hat{n}}$ satisfy
\begin{align*}
\mbox{\rm im}(V)& \subset \starspan_{i=1,...,m}\bigl\{X_{0}(\sigma_{i})B, \ldots,X_{k_{v}}(\sigma_{i})B \bigr\},\\
\mbox{\rm im}(W)& \subset \starspan_{i=1,...,m}\bigl\{X_{0}(2\sigma_{i})^{T}C^{T},\ldots, X_{k_{w}}(2\sigma_{i})^{T}C^{T}\bigr\},
\end{align*}
where $k_{v}=\max(P,Q)$ and $k_{w}=P+Q$. Then the reduced-order system ensures 
\begin{align}\label{pprtnf}
\begin{split}
\left.\frac{\partial^{p}}{\partial s_{1}^{p}} 
H_{1sym}(s_{1})\right|_{s_{1}=\sigma_{i}}
&=\left.\frac{\partial^{p}}{\partial s_{1}^{p}}
\hat{H}_{1sym}(s_{1})\right|_{s_{1}=\sigma_{i}}, \quad\enskip p=0,\ldots, k_v,\\
\left.\frac{\partial^{p}}{\partial s_{1}^{p}}
H_{1sym}(s_{1})\right|_{s_{1}=2\sigma_i}
&=\left.\frac{\partial^{p}}{\partial s_{1}^{p}}
\hat{H}_{1sym}(s_{1})\right|_{s_{1}=2\sigma_{i}}, \quad p=0,\ldots, k_w,\\
\left. \frac{\partial^{p+q}}{\partial s_{1}^{p}\partial s_{2}^{q}}
H_{2sym}(s_{1},s_{2})\right|_{ \stackon{$\scriptstyle s_{2}=\sigma_{i}$}{$\scriptstyle s_{1}=\sigma_{i}$}}
&=\left.\frac{\partial^{p+q}}{\partial s_{1}^{p}\partial s_{2}^{q}}
\hat{H}_{2sym}(s_{1},s_{2})\right|_{ \stackon{$\scriptstyle s_{2}=\sigma_{i}$}{$\scriptstyle s_{1}=\sigma_{i}$}},\;\; \begin{array}{l} p=0,\ldots, P,\\ q=0,\ldots, Q.\end{array}
\end{split}
\end{align}
\end{theorem}
For this general form, the proof is given in \cite{morYanJ18}. It is important to note that the construction of $V$ and $W$ is independent of the matrices $N$ and $H$. To clarify this, 
we consider the simple case $P=Q=1$ in more detail.

\noindent \textbf{Case $P=Q=1$:} The result in Theorem~\ref{thtnf} for this special case requires
\begin{align*}
\mbox{\rm im}(V)= \starspan\bigl\{&X_{0}(\sigma_{1})B,\ldots ,X_{0}(\sigma_{m})B,X_{1}(\sigma_{1})B,\ldots,
 X_{1}(\sigma_{m})B\bigr\},\\
\mbox{\rm im}(W)= \starspan \bigl\{&X_{0}(2\sigma_{1})^{T}C^{T},\ldots ,X_{0}(2\sigma_{m})^{T}C^{T},X_{1}(2\sigma_{1})^{T}C^{T},\ldots, X_{1}(2\sigma_{m})^{T}C^{T},\\&
X_{2}(2\sigma_{1})^{T}C^{T}\ldots, X_{2}(2\sigma_{m})^{T}C^{T}\bigr\},
\end{align*}
for which the following holds
\begin{align}\label{prtnf}
\begin{split}
\left.\frac{\partial^{p}}{\partial s_{1}^{p}}
H_{1sym}(s_{1})\right|_{s_{1}=\sigma_i}
&=\left.\frac{\partial^{p}}{\partial s_{1}^{p}}
\hat{H}_{1sym}(s_{1})\right|_{s_{1}=\sigma_{i}}, \quad\; p=0,1,\\
\left.\frac{\partial^{p}}{\partial s_{1}^{p}}H_{1sym}(s_{1})\right|_{s_{1}=2\sigma_i}
&=\left.\frac{\partial^{p}}{\partial s_{1}^{p}}\hat{H}_{1sym}(s_{1})\right|_{s_{1}=2\sigma_{i}}, \quad p=0,1,2,\\
\left.\frac{\partial^{p+q}}{\partial s_{1}^{p}\partial s_{2}^{q}}H_{2sym}(s_{1},s_{2})\right|_{ \stackon{$\scriptstyle s_{2}=\sigma_{i}$}{$\scriptstyle s_{1}=\sigma_{i}$}}
&=\left.\frac{\partial^{p+q}}{\partial s_{1}^{p}\partial s_{2}^{q}}\hat{H}_{2sym}(s_{1},s_{2})\right|_{ \stackon{$\scriptstyle s_{2}=\sigma_{i}$}{$\scriptstyle s_{1}=\sigma_{i}$}}, \quad p,q=0,1.
\end{split}
\end{align}
The proof is based on two important concepts that are crucial for implicit moment matching via projection. 
For any vector $X_j(\sigma_i)B$ in the image of $V$ 
	and any vector $X_j(2\sigma_i)^TC^T$ in the image of $W$, we have
\begin{equation}\label{case1eqs}
	V\hat{X}_{j}(\sigma_{i})\hat{B}=X_{j}(\sigma_{i})B, \qquad 
	W\hat{X}_{j}(2\sigma_{i})^T\hat{C}^T=X_{j}(2\sigma_{i})^TC^T.
\end{equation} 
Using these relations, it is easy to show that the first two equations in \eqref{prtnf} hold for $p=0$ and $q=0$. To see that these equations are also satisfied for the derivatives, we determine the first and second derivatives of $H_{1sym}(s_1)$.  They are given by
\begin{equation}
\begin{aligned}
\frac{\partial}{\partial s_1}H_{1sym}(s_1)&= -C(s_1E-A)^{-1}E(s_1E-A)^{-1}B = -CX_1(s_1)B,\\
\frac{\partial^2}{\partial s_1^2}H_{1sym}(s_1)&= 2C[(s_1E-A)^{-1}E]^{-2}(s_1E-A)^{-1}B = 2CX_2(s_1)B.
\end{aligned}
\end{equation}
Now premultiplying the first equation in \eqref{case1eqs} with $j=1$ by $C$ and the second equation in  \eqref{case1eqs} with $j=1,2$ by $B^T$, we observe that the derivatives of $H_{1sym}(s_1)$ and $\hat{H}_{1sym}(s_1)$ are also matched at $s_1=\sigma_i$ and $s_1=2\sigma_i$. 

For the third equation in \eqref{prtnf}, we have four different expressions depending on the values of $p$ and $q$. 
It follows from \eqref{hsym} that
\begin{equation*}
H_{2sym}(\sigma_{i},\sigma_{i}) = C X_{0}(2\sigma_{i})\Bigl(NX_{0}(\sigma_{i})B+ H\bigl(X_{0}(\sigma_{i})B\otimes X_{0}(\sigma_{i})B\bigr)\Bigr).
\end{equation*}
We can again use \eqref{case1eqs} to show that the third equation in \eqref{prtnf} holds for $p=q=0$. 
The partial derivative of $H_{2sym}(s_1,s_2)$ with respect to $s_1$ has the form
\begin{equation*} \label{case1partial}
\begin{aligned}
\frac{\partial}{\partial s_1}H_{2sym}(s_1,s_2) =&-\frac{1}{2!}\bigg[CX_{0}(s_1+s_2)\Bigl(NX_{1}(s_1)B+ H\bigl(X_{1}(s_1)B\otimes X_{0}(s_2)B  +X_{0}(s_2)B\otimes X_{1}(s_1)B\bigr)\Bigr) \\ &+ CX_{1}(s_1+s_2)\Bigl(N\bigl(X_{0}(s_1)B+X_0(s_2)B\bigr)  +H\bigl(X_{0}(s_1)B\otimes X_{0}(s_2)B \\ &+ X_{0}(s_2)B\otimes X_{0}(s_1)B\bigr) \Bigr)\bigg]. 
\end{aligned}
\end{equation*}
Then using \eqref{case1eqs} with $j=0,1$, we show that the third equation in \eqref{prtnf} holds for $p=1$ and $q=0$. Similarly, this equation can be proved for $p=0$ and $q=1$. 
For the final case $p=q=1$, we determine the partial derivative  
\begin{equation*} \label{case1partial2}
\small
\begin{aligned}
\frac{\partial^2}{\partial s_1\partial s_2}H_{2sym}(s_1,s_2) =&\frac{1}{2!}\bigg[CX_{0}(s_1\!+\!s_2)\Bigl(H\bigl(X_{1}(s_1)B\otimes X_{1}(s_2)B\! +\! X_{1}(s_2)B\otimes X_{1}(s_1)B\bigr)\Bigr)\\ 
& +CX_{1}(s_1+s_2)\Bigl(NX_{1}(s_2)B+ H\bigl(X_{0}(s_1)B\otimes X_{1}(s_2)B + X_{1}(s_2)B\otimes X_{0}(s_1)B\bigr)\Bigr)\\ 
& +CX_{1}(s_1+s_2)\Bigl(N X_{1}(s_1)B + H\bigl(X_{1}(s_1)B\otimes X_{0}(s_2)B + X_{0}(s_2)B\otimes X_{1}(s_1)B\bigr)\Bigr)\\ 
& +2CX_{2}(s_1+s_2)\Bigl(N\bigl(X_{0}(s_1)B+X_0(s_2)B\bigr)  + H\bigl(X_{0}(s_1)B\otimes X_{0}(s_2)B + X_{0}(s_2)B\otimes X_{0}(s_1)B\bigr) \Bigr)\bigg]. 
\end{aligned}
\end{equation*}
Using \eqref{case1eqs} with $j=0,1,2$, we see that the required multi-moment matching condition holds. Thus, \eqref{prtnf} holds for the specific choice of $V$ and $W$.

The result in Theorem~\ref{thtnf} utilises the symmetric form of the multivariate transfer functions, which is known to be complex especially for higher subsystems. 
In this paper, we propose a~multi-moment matching technique that utilises the regular form of the multivariate transfer functions \eqref{reg_tf123} and is, therefore, easy to be extended to higher multi-moments.

\section{Multi-Moment Matching for the First Two Regular Transfer Functions}
\label{sec:matching2}

In this section, we show how the choice of $V$ and $W$ can ensure generalised multi-moment matching implicitly. We begin with the concept of moment matching for the first two regular transfer functions of the quadratic-bilinear system \eqref{fom_qbdae} and later extend it to the first three regular transfer functions.
To solve Problem~\ref{prmulti} using the first two regular transfer functions given in \eqref{reg_tf123}, we first partition $H_{2reg}(s_{1},s_{2})$ additively into two parts
\begin{equation}\label{eq:H2reg}
H_{2reg}(s_{1},s_{2})=C\big(Z_{21}(s_{1},s_{2})+Z_{22}(s_{1},s_{2})\big),
\end{equation}
where
\begin{align}\label{I1I2}
\begin{split}
Z_{21}(s_{1},s_{2})&=X_{0}(s_{2})NX_{0}(s_{1})B,\\
Z_{22}(s_{1},s_{2})&=X_{0}(s_{2})H \bigl( X_{0}(s_{2}-s_{1})B\otimes X_{0}(s_{1})B \bigr).
\end{split}
\end{align}
With these notations, we can prove the following lemma.

\begin{lemma}\label{lemmapartial}
Let $Z_{21}(s_{1},s_{2})$ and $Z_{22}(s_{1},s_{2})$ be as defined in \eqref{I1I2}, 
then the partial derivatives of $H_{2reg}(s_{1},s_{2})$ can be written as 
$$
\frac{\partial^{p+q}}{\partial s_{1}^{p}\partial s_{2}^{q}}H_{2reg}(s_{1},s_{2})=C\bigg(\frac{\partial^{p+q}}{\partial s_{1}^{p}\partial s_{2}^{q}}Z_{21}(s_{1},s_{2})+\frac{\partial^{p+q}}{\partial s_{1}^{p}\partial s_{2}^{q}}Z_{22}(s_{1},s_{2})\bigg),
$$
where
\begin{align}
\frac{\partial^{p+q}}{\partial s_{1}^{p}\partial s_{2}^{q}}Z_{21}(s_{1},s_{2})&=(-1)^{p+q}p!q!X_{q}(s_{2}) N X_{p}(s_{1})B,\label{eq:Z1}\\
\frac{\partial^{p+q}}{\partial s_{1}^{p}\partial s_{2}^{q}}Z_{22}(s_{1},s_{2})&= \sum_{j=0}^{q}\begin{pmatrix}q\\ j\end{pmatrix}(q-j)!\, X_{q-j}(s_{2})   
H \Bigl( \sum_{k=0}^{p}(-1)^{p+q-k}\label{eq:Z2}\\ & \times \begin{pmatrix}p\\ k\end{pmatrix} (k+j)!X_{k+j}(s_{2}-s_{1})B \otimes (p-k)! X_{p-k}(s_{1})B \Bigr). \nonumber
\end{align}
\end{lemma}

\begin{proof}  
Equation \eqref{eq:Z1} immediately follows from \eqref{lemma_eq1} with $j=0$ and $\ell=p,q$. 
In order to prove \eqref{eq:Z2}, we first observe that
$$
\frac{\partial^{p} }{\partial s_{1}^p}Z_{22}(s_{1},s_{2}) = 
X_{0}(s_{2})H \Bigl(\sum_{k=0}^p (-1)^{p-k} \begin{pmatrix}p\\ k\end{pmatrix} i!\,X_{k}(s_{2}-s_{1})B \otimes (p-k)! X_{p-k}(s_{1})B\Bigr).
$$
This relation is obtain by using the Leibniz product rule for the $p$-th partial deri\-vative with respect to $s_1$ combined with \eqref{lemma_eq1} with $j=0$ and $\ell=k,p-k$. Applying again the Leibniz product rule for the $q$-th partial derivative with respect to $s_2$ to the above equation, we get equation~\eqref{eq:Z2}. 		

\end{proof}


Next, we show how a reduced system can be identified that can match the multi-moments in Lemma~\ref{lemmapartial} for the first two regular transfer functions without using the system matrices $N$ and $H$. 

\begin{theorem} \label{theorem2}
For a~given set of interpolation points $\sigma_{i}\notin \Lambda(A,E)\cup\Lambda(\hat{A},\hat{E})$, 
$i=1,\ldots, m$, let $V=V(\sigma_i,P_1)$, ${W=W(\sigma_i,Q) \in \mathbb{R}^{n\times \hat{n}}}$ satisfy
\begin{align*}
&\mbox{\rm im}(V)= \starspan\bigl\{X_{0}(\sigma_{1})B,\ldots ,X_{0}(\sigma_{m})B,\ldots,X_{P_1}(\sigma_{1})B,\ldots, X_{P_1}(\sigma_{m})B\bigr\},\\
&\mbox{\rm im}(W)= \starspan\bigl\{\!X_{0}(2\sigma_{1})^{T}\!C^{T}\!,\ldots\!,X_{0}(2\sigma_{m})^{T}\!C^{T},\!\ldots,\!X_{Q}(2\sigma_{1})^{T}\!C^{T}\!,\ldots, X_{Q}(2\sigma_{m})^{T}\!C^{T}\! \bigr\}
\end{align*}
with $P_1=P+Q$. Then the reduced-order system satisfies
\begin{align}\label{prprop}
\begin{split}
\left.\frac{\partial^{p}}{\partial s_{1}^{p}} H_{1reg}(s_{1})\right|_{s_{1}=\sigma_{i}}
&=\left.\frac{\partial^{p}}{\partial s_{1}^{p}}\hat{H}_{1reg}(s_{1})\right|_{s_{1}=\sigma_{i}}, \qquad p=0\ldots,P+Q,\\
\left.\frac{\partial^{q}}{\partial s_{1}^{q}}H_{1reg}(s_{1})\right|_{s_{1}=2\sigma_i}
&=\left.\frac{\partial^{q}}{\partial s_{1}^{q}}\hat{H}_{1reg}(s_{1})\right|_{s_{1}=2\sigma_{i}}, \qquad q=0,\ldots, \max(P,Q),\\
\left.\frac{\partial^{p+q}}{\partial s_{1}^{p}\partial s_{2}^{q}}H_{2reg}(s_{1},s_{2})\right|_{ \stackon{$\scriptstyle s_{2}=2\sigma_{i}$}{$\scriptstyle s_{1}=\sigma_{i}\enskip$}}
&=\left.\frac{\partial^{p+q}}{\partial s_{1}^{p}\partial s_{2}^{q}}\hat{H}_{2reg}(s_{1},s_{2})\right|_{ \stackon{$\scriptstyle s_{2}=2\sigma_{i}$}{$\scriptstyle s_{1}=\sigma_{i}\enskip$}}, \quad \begin{array}{l} p=0,\ldots, P, \\ q=0,\ldots, Q.\end{array}
\end{split}
\end{align}
\end{theorem}

\begin{proof}
Since $H_{1reg}(s_1)=H_{1sym}(s_1)$, the proof for the first two equations in \eqref{prprop} is similar to Theorem~\ref{thtnf}. In order to show the third equation in \eqref{prprop}, we consider the additive decomposition of $H_{2reg}(s_1,s_2)$ in \eqref{eq:H2reg}. Using Lemma~\ref{lemmapartial} and \eqref{case1eqs}, we have
\begin{align*}
\left.C\,\frac{\partial^{p+q}}{\partial s_{1}^{p}\partial s_{2}^{q}}Z_{21}(s_{1},s_{2})\right|_{ \stackon{$\scriptstyle s_{2}=2\sigma_{i}$}{$\scriptstyle s_{1}=\sigma_{i}\enskip$}} 
&=(-1)^{p+q}p!q!CX_{q}(2\sigma_{i}) NX_{p}(\sigma_{i})B\\
&=(-1)^{p+q}p!q!\hat{C}\hat{X}_{q}(2\sigma_{i}) W^{T}NV\hat{X}_{p}(\sigma_{i})\hat{B}\\
&=(-1)^{p+q}p!q!\hat{C}\hat{X}_{q}(2\sigma_{i})\hat{N}\hat{X}_{p}(\sigma_{i})\hat{B}\\ 
&=\left.\hat{C}\,\frac{\partial^{p+q}}{\partial s_{1}^{p}\partial s_{2}^{q}}\hat{Z}_{21}(s_{1},s_{2})\right|_{ \stackon{$\scriptstyle s_{2}=2\sigma_{i}$}{$\scriptstyle s_{1}=\sigma_{i}\enskip$}}
\end{align*}
and
\begin{align*}
\left.C\,\frac{\partial^{p+q}}{\partial s_{1}^{p}\partial s_{2}^{q}}Z_{22}(s_{1},s_{2})\right|_{ \stackon{$\scriptstyle s_{2}=2\sigma_{i}$}{$\scriptstyle s_{1}=\sigma_{i}\enskip$}}  
	&= C\,\sum_{j=0}^{q} \begin{pmatrix} q\\ j\end{pmatrix}(q-j)!\, X_{q-j}(2\sigma_{i})  \sum_{k=0}^{p} (-1)^{p+q-k}
 \begin{pmatrix} p\\ k\end{pmatrix} \\
& \qquad \times  H\bigl((k+j)!X_{k+j}(\sigma_{i})B\otimes (p-k)! X_{p-k}(\sigma_{i})B \bigr)\\
& =\hat{C}\,\sum_{j=0}^{q}  \begin{pmatrix} q\\ j\end{pmatrix} (q-j)!\, \hat{X}_{q-j}(2\sigma_{i})W^{T} 
\sum_{k=0}^{p} (-1)^{p+q-k} \begin{pmatrix} p\\ k\end{pmatrix}\\ 
& \qquad \times H\bigl((k+j)!V\hat{X}_{k+j}(\sigma_{i})\hat{B}\otimes (p-k)! V\hat{X}_{p-k}(\sigma_{i})\hat{B} \bigr)\\
& = \hat{C}\,\sum_{j=0}^{q}  \begin{pmatrix} q\\ j\end{pmatrix} (q-j)!\,\hat{X}_{q-j}(2\sigma_{i}) 
\sum_{k=0}^{p} (-1)^{p+q-k}  \begin{pmatrix} p\\ k\end{pmatrix}\\  
& \qquad \times \hat{H}\bigl((k+j)!\hat{X}_{k+j}(\sigma_{i})\hat{B}\otimes (p-k)! \hat{X}_{p-k}(\sigma_{i})\hat{B} \bigr)\\  
& = \left.\hat{C}\,\frac{\partial^{p+q}}{\partial s_{1}^{p}\partial s_{2}^{q}}\hat{Z}_{22}(s_{1},s_{2})\right|_{ \stackon{$\scriptstyle s_{2}=2\sigma_{i}$}{$\scriptstyle s_{1}=\sigma_{i}\enskip$}}.
\end{align*}
Combining these equalities, we get the third relation in \eqref{prprop}.
\end{proof}

Based on Theorem~\ref{theorem2},  we propose the multi-moment matching method for model reduction of the quadratic-bilinear system \eqref{fom_qbdae} as presented in Algorithm~\ref{alg:1}.
 
\begin{algorithm}[h] 
\SetAlgoLined
\KwIn{$E, A, N,H\in\mathbb{R}^{n\times n}$, $B, C^T\in\mathbb{R}^n$, $\sigma_{i} \in \mathbb{C}$ for 
$i=1,\ldots,m$, $P, Q\in\mathbb{N}$.}
\KwOut{$\hat{E}, \hat{A},\hat{N}, \hat{H}\in\mathbb{R}^{\hat{n}\times \hat{n}}$, $\hat{B}, \hat{C}^T\in\mathbb{R}^{\hat{n}}$.}
$V=[~]$; $W=[~]$;\\
\For{$j=0 : P+Q $}{
\For{$i=1 : m$}{$V=[V,~ X_{j}(\sigma_{i})B]$}
}
\For{$j=0 : Q$}{
\For{$i=1 : m$}{$W=[W,~ X_{j}(2\sigma_{i})^{T}C^T]$}
  }
$U=orth([V,W])$, \\
Construct the reduced-order matrices 
\begin{align*}
\hat{E}&=U^{T}EU,& \hat{A}&=U^{T}AU,& \hat{B}&=U^{T}B,&\\
\hat{C}&=CU,& \hat{N}&=U^{T}NU,& \hat{H}&=U^{T}H(U\otimes U).
\end{align*}
\caption{Multi-moment matching using first two regular transfer functions}
\label{alg:1}
\end{algorithm}

\section{Multi-Moment Matching for the First Three Regular Transfer Functions}
\label{sec:matching3}

In this section, we extend the concept of higher-order moment matching to the third regular transfer function $H_{3reg}(s_1,s_2,s_3)$. This is possible because the structure of the third regular subsystem is relatively simple compared to the corresponding symmetric form.  Before proceeding to the third regular transfer function, we introduce a~new function 
$$
Z(s_{1},s_{2})= NX_{0}(s_{1})B+H ( X_{0}(s_{2}-s_{1})B\otimes X_{0}(s_{1})B),
$$ 
which implies
$$
X_0(s_2)Z(s_{1},s_{2})=Z_{21}(s_{1},s_{2})+ Z_{22}(s_{1},s_{2}).
$$
Then the third regular transfer functions can be written as  
\begin{align*}
\begin{split}
H_{3reg}(s_{1},s_{2},s_{3}) &=
CX_{0}(s_{3})\Bigl(NX_0(s_2)Z(s_{1},s_{2}) +H\bigl(X_0(s_{3}-s_{2})B \otimes X_0(s_2)Z(s_{1},s_{2})\\
& \quad +X_0(s_{3}-s_{1})Z(s_{2}-s_{1},s_{3}-s_{1})\otimes X_0(s_{1})B\bigr)\Bigr).
\end{split}
\end{align*}
It can be additively partitioned as 
\begin{align*}
H_{3reg}(s_{1},s_{2},s_{3})&=C\bigl(Z_{31}(s_{1},s_{2},s_{3}) + Z_{32}(s_{1},s_{2},s_{3})+Z_{33}(s_{1},s_{2},s_{3})\bigr),
\end{align*}
where
\begin{equation} \label{I1I2I3}
\begin{aligned}
Z_{31}(s_{1},s_{2},s_{3}) &= X_{0}(s_{3})NX_0(s_{2})Z(s_{1},s_{2}), \\
Z_{32}(s_{1},s_{2},s_{3}) &= X_{0}(s_{3})H\bigl(X_0(s_{3}-s_{2})B \otimes X_0(s_{2})Z(s_{1},s_{2})\bigr), \\ 
Z_{33}(s_{1},s_{2},s_{3}) &=X_{0}(s_{3})H\bigl(X_0(s_{3}-s_{1})Z(s_{2}-s_{1},s_{3}-s_{1})\otimes X_0(s_{1})B\bigr).
\end{aligned}
\end{equation}
With these notations, we can extend Lemma~\ref{lemmapartial} to $H_{3reg}(s_{1},s_{2},s_{3})$.

\begin{lemma}\label{lemmapartialk3}
Let $Z_{3j}(s_{1},s_{2},s_{3})$, $j=1,2,3$,
be as defined in \eqref{I1I2I3}. Then the partial derivatives of $H_{3reg}(s_{1},s_{2},s_{3})$ can be written as 
\begin{align*}
\frac{\partial^{p+q+\ell}}{\partial s_{1}^{p}\partial s_{2}^{q}\partial s_{3}^{z}}H_{3reg}(s_{1},s_{2},s_3)=&C\Bigl(\frac{\partial^{p+q+\ell}}{\partial s_{1}^{p}\partial s_{2}^{q}\partial s_{3}^{\ell}}Z_{31}(s_{1},s_{2},s_3)\\ &+\frac{\partial^{p+q+\ell}}{\partial s_{1}^{p}\partial s_{2}^{q}\partial s_{3}^{\ell}}Z_{32}(s_{1},s_{2},s_3) +\frac{\partial^{p+q+\ell}}{\partial s_{1}^{p}\partial s_{2}^{q}\partial s_{3}^{\ell}}Z_{33}(s_{1},s_{2},s_3)\Bigr),
\end{align*}
where
\begin{align*}
\begin{split}
&\frac{\partial^{p+q+\ell}}{\partial s_{1}^{p} \partial s_{2}^{q} \partial s_{3}^{\ell}}Z_{31}(s_{1},s_{2},s_{3})
=(-1)^{\ell}\ell!\,X_{\ell}(s_{3}) N\frac{\partial^{p+q}}{\partial s_{1}^{p}\partial s_{2}^{q}}\bigg(Z_{21}(s_{1},s_{2}) +
Z_{22}(s_{1},s_{2})\bigg),\\
&\frac{\partial^{p+q+\ell}}{\partial s_{1}^{p} \partial s_{2}^{q} \partial s_{3}^{\ell}}Z_{32}(s_{1},s_{2},s_{3})=\sum_{k_{\ell}=0}^{\ell}
 (-1)^\ell\begin{pmatrix}\ell\\k_\ell\end{pmatrix}
(\ell-k_{\ell})!X_{\ell-k_{\ell}}(s_{3}) H \sum_{k_{q}=0}^{q}
\begin{pmatrix}q\\k_q\end{pmatrix} \\
&\enskip \times \bigg[(k_{\ell}  +k_{q})!X_{k_{q}+k_{\ell}}(s_{3}-s_{2})B  \otimes
\frac{\partial^{p+q-k_{q}}}{\partial s_{1}^{p} \partial s_{2}^{q-k_{q}}}
\Bigl(Z_{21}(s_{1},s_{2})  +
Z_{22}(s_{1},s_{2})\Bigr) \bigg],\\
&\frac{\partial^{p+q+\ell}}{\partial s_{1}^{p} \partial s_{2}^{q} \partial s_{3}^{\ell}} Z_{33}(s_{1},s_{2},s_{3})=\sum_{k_{\ell}=0}^{\ell}\begin{pmatrix}\ell\\k_\ell\end{pmatrix} (\ell-k_{\ell})!X_{\ell-k_{\ell}}(s_{3}) H
  \sum_{k_{p}=0}^{p}(-1)^{p+\ell-k_{p}-k_{\ell}}\begin{pmatrix}p\\k_p\end{pmatrix}\\
&\enskip \times\! \!\bigg[
\frac{\partial^{q+k_{p}+k_{\ell}}}{\partial s_{1}^{k_{p}}\partial s_{2}^{q}\partial s_{3}^{k_{\ell}}}  \!
\Bigl(Z_{21}\!(s_{2}\!-\!s_{1},s_{3}\!-\!s_{1})\!+\! Z_{22}\!(s_{2}\!-\!s_{1},s_{3}\!-\!s_{1})\Bigr) \!\otimes\! (p\!-\!k_{p})!X_{p-k_{p}}(s_{1})B\bigg]. 
\end{split}
\end{align*}
\end{lemma}

\begin{proof}
The proof is similar to that of Lemma~\ref{lemmapartial} and, therefore, omitted.
\end{proof}

Introducing the functions
\begin{align*}
R_N(\sigma_i,p) &= (-1)^p p!\, N X_p(s_1)B, \\	
R_H(\sigma_i,p,j) &= H\bigg(\sum_{k=0}^{p}(-1)^{p-k+j} \begin{pmatrix} p\\ k\end{pmatrix}
	(k+j)!\, X_{k+j}(\sigma_i)B \otimes (p-k)!X_{p-k}(\sigma_i)B\bigg), 
\end{align*} 
the partial derivatives of $Z_{21}(s_{1},s_{2})+ Z_{22}(s_{1},s_{2})$ can shorty be written as  
\begin{align*}
 \frac{\partial^{p+q}}{\partial s_1^p \partial s_2^q }\left.\big(Z_{21}(s_{1},s_{2})+ Z_{22}(s_{1},s_{2})\big)\right|_{\stackon{$\scriptstyle s_2=2\sigma_{i}$}{$\scriptstyle s_{1}=\sigma_{i}\enskip$}} =&(-1)^{q} q!\,X_{q}(2\sigma_i)R_N(\sigma_i,p) \\ & +
 \sum_{j=0}^{q} (-1)^{q-j} \begin{pmatrix} q\\ j\end{pmatrix} (q-j)!\, X_{q-j}(2\sigma_i)R_H(\sigma_i,p,j).
\end{align*}
With these observations, we can prove the following lemma. 

\begin{lemma}\label{lemmak3V}
For a~given set of interpolation points $\sigma_{i}\notin \Lambda(A,E)\cup\Lambda(\hat{A},\hat{E})$,
\mbox{$i=1,\ldots, m$}, let $V_1\in \mathbb{R}^{n\times r}$ satisfy 
\begin{align*}
\mbox{\rm im}(V_1)&= \starspan_{i=1,\ldots,m}\bigl\{V(\sigma_i,P_{1}), V_N(\sigma_i,0,0)+V_H(\sigma_i,0,0),\ldots,V_N(\sigma_i,P_{2},Q_{2})+V_H(\sigma_i,P_{2},Q_{2})\bigr\}, 
\end{align*}
where $V(\sigma_i,P_1)$ is defined as in Theorem~\ref{theorem2} with $P_{1}=P+Q$ and  
\begin{align*}
V_N(\sigma_i,P_{2},Q_{2})=& \biggl[X_0(2\sigma_i)R_N(\sigma_i,0),\ldots,X_{0}(2\sigma_i)R_N(\sigma_i,P_{2}),\ldots,\\&(-1)^{Q_{2}}Q_{2}!X_{Q_{2}}(2\sigma_i)R_N(\sigma_i,0),\ldots,(-1)^{Q_{2}}Q_{2}!X_{Q_{2}}(2\sigma_i)R_N(\sigma_i,P_{2})\biggr],\\
V_H(\sigma_i,P_{2},Q_{2})=& \biggl[X_0(2\sigma_i)R_H(\sigma_i,0,0),\big[X_0(2\sigma_i)R_H(\sigma_i,0,1)-X_1(2\sigma_i)R_H(\sigma_i,0,0)\big],\ldots,\\ &\big[X_0(2\sigma_i)R_H(\sigma_i,0,Q_{2})+ \ldots +(-1)^{Q_{2}}Q_{2}! X_{Q_{2}}(2\sigma_i)R_H(\sigma_i,0,0)\big],\ldots,\\&
X_0(2\sigma_i)R_H(\sigma_i,P_{2},0),\big[X_0(2\sigma_i)R_H(\sigma_i,P_{2},1)-X_1(2\sigma_i)R_H(\sigma_i,P_{2},0)\big],\ldots,\\ &\big[X_0(2\sigma_i)R_H(\sigma_i,P_{2},Q_{2})+ \ldots +(-1)^{Q_{2}}Q_{2}! X_{Q_{2}}(2\sigma_i)R_H(\sigma_i,P_{2},0)\big]\biggr],
\end{align*}
with $P_{2}=P$ and $Q_{2}=Q$.  Furthermore, let $W$ be as in Theorem~\ref{theorem2}. Then the reduced-order system obtained using the projection matrices $V_1$ and $W$ satisfies 
\begin{equation}\label{leeq}
\begin{aligned} 
\left. V_1\frac{\partial^{p+q}}{\partial s_1^p \partial s_2^q}(\hat{Z}_{21}(s_{1},s_{2})+\hat{Z}_{22}(s_{1},s_{2}))\right|_{\stackon{$\scriptstyle s_2=2\sigma_{i}$}{$\scriptstyle s_{1}=\sigma_{i}\enskip$}}   =\left. \frac{\partial^{p+q}}{\partial s_1^p \partial s_2^q}(Z_{21}(s_{1},s_{2})+Z_{22}(s_{1},s_{2}))\right|_{\stackon{$\scriptstyle s_2=2\sigma_{2}$}{$\scriptstyle s_{1}=\sigma_{i}\enskip$}}
\end{aligned}
\end{equation}
for $p=0,\ldots,P$ and $q=0,\ldots,Q$.
\end{lemma}

\begin{proof}
The left-hand side of \eqref{leeq} can be written as
\begin{align*}
&V_1 \left.\frac{\partial^{p+q}}{\partial s_{1}^{p}\partial s_{2}^{q}} \bigg(\hat{Z}_{21}(s_{1},s_{2})+ \hat{Z}_{22}(s_{1},s_{2})\bigg)\right|_{\stackon{$\scriptstyle s_2=2\sigma_{i}$}{$\scriptstyle s_{1}=\sigma_{i}\enskip$}} 
\\ &\quad =  V_1\Bigl((-1)^{q}q!\hat{X}_{q}(2\sigma_i)W^TR_N(\sigma_i,p) + 
\sum_{j=0}^{q}(-1)^{q-j} \begin{pmatrix}q\\j\end{pmatrix} (q-j)!\hat{X}_{q-j}(2\sigma_i)W^TR_H(\sigma_i,p,j)\Bigr)
\\ &\quad =  (-1)^{q}q!X_{q}(2\sigma_i)R_N(\sigma_i,p)  +\sum_{j=0}^{q}(-1)^{q-j}\begin{pmatrix}q\\j\end{pmatrix}(q-j)!X_{q-j}(2\sigma_i)R_H(\sigma_i,p,j).
\end{align*}
This proves the lemma where the last equality follows since $V_{N}(\sigma_i,P,Q)+V_{H}(\sigma_i,P,Q)$ lies
in the image of $V_1$ for all values of $p$ and $q$.
\end{proof}

\begin{remark} Similarly to Lemma~\ref{lemmak3V}, one can also show the relation 
\begin{align*}
&\left.V_{1}\frac{\partial^{\tau}}{\partial s^{\tau}} \bigg(\hat{Z}_{21}(f_{1}(s),f_{2}(s))+ \hat{Z}_{22}(f_{1}(s),f_{2}(s))\bigg)\right|_{\stackon{$\scriptstyle f_{2}(s)=2\sigma_{i}$}{$\scriptstyle f_{1}(s)=\sigma_{i}\enskip$}}\\&\qquad =\left.\frac{\partial^{\tau}}{\partial s^{\tau}} \bigg(Z_{21}(f_{1}(s),f_{2}(s))+ Z_{22}(f_{1}(s),f_{2}(s))\bigg)\right|_{\stackon{$\scriptstyle f_{2}(s)=2\sigma_{i}$}{$\scriptstyle f_{1}(s)=\sigma_{i}\enskip$}}
\end{align*}
for any $f_{1}(s)$ and $f_{2}(s)$, where $s$ is the vector of variables $s_{1},s_{2},\ldots$. This will change the values of $P_{1}$, $P_{2}$ and $Q_{2}$ depending on the number of variables, the order of their derivatives and the combination of these variables used. In particular, this relation is useful in the regular form with $Z_{33}(s_{1},s_{2},s_{3})$, where $f_{1}(s)=s_{2}-s_{1}$ and $f_{2}(s)=s_{3}-s_{1}$.
\end{remark}

We can now identify the structure of the projection matrices $V_1$ and $W_1$ providing a reduced-order model which matches the multi-moments of the first three regular transfer functions while involving non-linear matrices $N$ and $H$. 

\begin{theorem} \label{theorem3}
For a~given set of interpolation points $\sigma_{i}\notin \Lambda(A,E)\cup\Lambda(\hat{A},\hat{E})$,
$i=1,\ldots, m$, let $V_1\in \mathbb{R}^{n\times r}$ be as defined in Lemma~\ref{lemmak3V} with 
$P_{1}=P+\max(Q,L)$, $P_{2}=\max(Q, P+L)$ and $Q_{2}=P+Q$, and let $W_1\in \mathbb{R}^{n\times r}$ satisfy 
\begin{align*}
\mbox{\rm im}(W_1)= \starspan\bigl\{W, X_{0}(3\sigma_{i})^{T}\!C^{T}\!,\ldots,X_{L}(3\sigma_{i})^{T}\!C^{T}\!, X_{0}(3\sigma_m)^{T}\!C^{T}\!, \ldots, X_{L}(3\sigma_{m})^{T}\!C^{T}\! \bigr\},
\end{align*}
where $W$ is as in Theorem~\ref{theorem2}. Then the reduced-order system satisfies
\begin{align}\label{regG}
\begin{split}
\left.\frac{\partial^{p}}{\partial s_{1}^{p}} H_{1reg}(s_{1})\right|_{s_{1}=\sigma_{i}}
&\!=\left.\frac{\partial^{p}}{\partial s_{1}^{p}}\hat{H}_{1reg}(s_{1})\right|_{s_{1}=\sigma_{i}}, \quad p=0,\ldots, P_1,\\
\left.\frac{\partial^{q}}{\partial s_{1}^{q}}H_{1reg}(s_{1})\right|_{s_{1}=2\sigma_i}
&\!=\left.\frac{\partial^{q}}{\partial s_{1}^{q}}\hat{H}_{1reg}(s_{1})\right|_{s_{1}=2\sigma_{i}}, \quad q=0,\ldots,Q_2,\\
\left.\frac{\partial^{\ell}}{\partial s_{1}^{\ell}}H_{1reg}(s_{1})\right|_{s_{1}=3\sigma_i}
&\!=\left.\frac{\partial^{\ell}}{\partial s_{1}^{\ell}}\hat{H}_{1reg}(s_{1})\right|_{s_{3}=3\sigma_{i}}, \quad \ell=0,\ldots, L,\\
\left.\frac{\partial^{p+q}}{\partial s_{1}^{p}\partial s_{2}^{q}}H_{2reg}(s_{1},s_{2})\right|_{ \stackon{$\scriptstyle s_{2}=2\sigma_{i}$}{$\scriptstyle s_{1}=\sigma_{i}\enskip$}}
&\!=\left.\frac{\partial^{p+q}}{\partial s_{1}^{p}\partial s_{2}^{q}}\hat{H}_{2reg}(s_{1},s_{2})\right|_{ \stackon{$\scriptstyle s_{2}=2\sigma_{i}$}{$\scriptstyle s_{1}=\sigma_{i}\enskip$}}\!,\;
\begin{array}{l} p=0,\ldots, P_2, \\ q=0,\ldots, Q_2,\end{array}
\end{split}
\end{align}
\begin{align*}
\left.\frac{\partial^{p+q+\ell}}{\partial s_{1}^{p}\partial s_{2}^{q}\partial s_{3}^{\ell}}\hat{H}_{3reg}(s_{1},s_{2},s_{3})\right|_{\Shortstack{$\scriptstyle s_{1}=\sigma_{i}\enskip$ $\scriptstyle s_{2}=2\sigma_{i}$ $\scriptstyle s_{3}=3\sigma_{i}$}} 
&\!= \left. \frac{\partial^{p+q+\ell}}{\partial s_{1}^{p}\partial s_{2}^{q}\partial s_{3}^{\ell}}H_{3reg}(s_{1},s_{2},s_{3})\right|_{\Shortstack{$\scriptstyle s_{1}=\sigma_{i}\enskip$ $\scriptstyle s_{2}=2\sigma_{i}$ $\scriptstyle s_{3}=3\sigma_{i}$}},\\
 p=0,\ldots,  P,\quad &q=0,\ldots, Q,\quad \ell=0,\ldots,L.
\end{align*}
\end{theorem}

\begin{proof}
The first four equations in \eqref{regG} are easy to prove using Lemma~\ref{lemmak3V} and Theorem~\ref{theorem2}. To show the last equation, we begin with
\begin{equation*}
\begin{aligned}
\hat{C}\frac{\partial^{p+q+\ell}}{\partial s_{1}^{p}\partial s_{2}^{q}\partial s_{3}^{\ell}}& \left. \hat{Z}_{31}(s_{1},s_{2},s_{3})\right|_{\substack{s_{1}=\sigma_{i}\enskip \\ s_{2}=2\sigma_{i} \\ s_{3}=3\sigma_{i}}} \\
& = \hat{C}(-1)^{\ell}\ell!\hat{X}_{\ell}(3\sigma_{i})\hat{N}\left.
\frac{\partial^{p+q}}{\partial s_{1}^{p}\partial s_{2}^{q}}\bigg(\hat{Z}_{21}(s_{1},s_{2})
+ \hat{Z}_{22}(s_{1},s_{2})\bigg)\right|_{\substack{s_{1}=\sigma_{i}\enskip \\s_{2}=2\sigma_{i}}} \\ 
& = (-1)^{\ell}\ell!\hat{C}\hat{X}_{\ell}(3\sigma_{i})W_1^{T}NV_1
\frac{\partial^{p+q}}{\partial s_{1}^{p}\partial s_{2}^{q}}\bigg(\hat{Z}_{21}(s_{1},s_{2})  
+\hat{Z}_{22}(s_{1},s_{2})\left.\bigg)\right|_{\substack{s_{1}=\sigma_{i}\enskip\\s_{2}=2\sigma_{i}}} \\ 
& = (-1)^{\ell}\ell!CX_{\ell}(3\sigma_{i})N
\left.\frac{\partial^{p+q}}{\partial s_{1}^{p}\partial s_{2}^{q}}\bigg(Z_{21}(s_{1},s_{2})
+Z_{22}(s_{1},s_{2})\bigg)\right|_{\substack{s_2=2\sigma_{i} \\ s_{1}=\sigma_i}}\\ 
& = \left.C\frac{\partial^{p+q+\ell}}{\partial s_{1}^{p}\partial s_{2}^{q}\partial s_{3}^{\ell}}Z_{31}(s_{1},s_{2},s_{3})\right|_{\substack{s_{1}=\sigma_{i}\enskip\\ s_{2}=2\sigma_{i} \\ s_{3}=3\sigma_{i}}},
\end{aligned}
\end{equation*}
where the second last equation is due to Lemma \ref{lemmak3V}. Now for $\hat{Z}_{32}(s_{1},s_{2},s_{3})$, we can write 
\begin{align*}
 &\hat{C} \left.\frac{\partial^{p+q+\ell}}{\partial s_{1}^{p}\partial s_{2}^{q}\partial s_{3}^{\ell}}\hat{Z}_{32}(s_{1},s_{2},s_{3})\right|_{\substack{ s_{1}=\sigma_{i}\enskip\\ s_{2}=2\sigma_{i} \\s_{3}=3\sigma_{i}}}\\ &
= \hat{C} \sum_{k_{\ell}=0}^{\ell}(-1)^{\ell-k_{\ell}}\begin{pmatrix}\ell\\k_\ell\end{pmatrix}(\ell-k_{\ell})!\hat{X}_{\ell-k_{\ell}}(3\sigma_{i}) \hat{H}
  \sum_{k_{q}=0}^{q} \Bigg((-1)^{k_{q}+k_{\ell}}\begin{pmatrix}q\\k_q\end{pmatrix}(k_{q}+k_{\ell})!\hat{X}_{k_{q}+k_{\ell}}(\sigma_{i})\hat{B}
  \\&
 \otimes \bigg[\left. \frac{\partial^{p+q-k_{q}}}{\partial s_{1}^{p}\partial s_{2}^{q-k_{q}}}\hat{Z}_{21}(s_{1},s_{2})\right|_{\Shortstack{ $\scriptstyle s_{1}=\sigma_{i}\enskip$ $\scriptstyle s_{2}=2\sigma_{i}$}}  +\left. \frac{\partial^{p+q-k_{q}}}{\partial s_{1}^{p}\partial s_{2}^{q-k_{q}}}\hat{Z}_{22}(s_{1},s_{2})\right|_{\Shortstack{$\scriptstyle s_{1}=\sigma_{i}\enskip$ $\scriptstyle s_{2}=2\sigma_{i}$}}\bigg] \Bigg) 
 \\ 
= \hat{C} &\sum_{k_{\ell}=0}^{\ell}(-1)^{\ell-k_{\ell}}\begin{pmatrix}\ell\\k_\ell\end{pmatrix}(\ell-k_{\ell})!\hat{X}_{\ell-k_{\ell}}(3\sigma_{i}) 
 W^{T}H(V_{1}\otimes V_{1}) \sum_{k_{q}=0}^{q} \Bigg((-1)^{k_{q}+k_{\ell}}\begin{pmatrix}q\\k_q\end{pmatrix}(k_{q}+k_{\ell})!\hat{X}_{k_{q}+k_{\ell}}(\sigma_{i})\hat{B}
  \\  &
 \otimes \bigg[\left.\frac{\partial^{p+q-k_{q}}}{\partial s_{1}^{p}\partial s_{2}^{q-k_{q}}}\hat{Z}_{21}(s_{1},s_{2})\right|_{\Shortstack{$\scriptstyle s_{1}=\sigma_{i}\enskip$ $\scriptstyle s_{2}=2\sigma_{i}$}}  +\left.\frac{\partial^{p+q-k_{q}}}{\partial s_{1}^{p}\partial s_{2}^{q-k_{q}}}\hat{Z}_{22}(s_{1},s_{2})\right|_{\Shortstack{$\scriptstyle s_{1}=\sigma_{i}\enskip$ $\scriptstyle s_{2}=2\sigma_{i}$}}\bigg] \Bigg) 
 \\ 
= \hat{C} &\sum_{k_{\ell}=0}^{\ell}(-1)^{\ell-k_{\ell}}\begin{pmatrix}\ell\\k_\ell\end{pmatrix}(\ell-k_{\ell})!\hat{X}_{\ell-k_{\ell}}(3\sigma_{i}) 
  W^{T}H\sum_{k_{q}=0}^{q} \Bigg(V_{1}(-1)^{k_{q}+k_{\ell}}\begin{pmatrix}q\\k_q\end{pmatrix}(k_{q}+k_{\ell})!\hat{X}_{k_{q}+k_{\ell}}(\sigma_{i})\hat{B}
  \\   &
 \otimes V_{1} \bigg[\left. \frac{\partial^{p+q-k_{q}}}{\partial s_{1}^{p}\partial s_{2}^{q-k_{q}}}\hat{Z}_{21}(s_{1},s_{2})\right|_{\Shortstack{ $\scriptstyle s_{1}=\sigma_{i}\enskip$ $\scriptstyle s_{2}=2\sigma_{i}$}}  +\left. \frac{\partial^{p+q-k_{q}}}{\partial s_{1}^{p}\partial s_{2}^{q-k_{q}}}\hat{Z}_{22}(s_{1},s_{2})\right|_{\Shortstack{$\scriptstyle s_{1}=\sigma_{i}\enskip$ $\scriptstyle s_{2}=2\sigma_{i}$}}\bigg] \Bigg) \\
= C&\sum_{k_{\ell}=0}^{\ell}(-1)^{\ell-k_{\ell}}\begin{pmatrix}\ell\\k_\ell\end{pmatrix}(\ell-k_{\ell})!X_{\ell-k_{\ell}}(3\sigma_{i}) 
  H\sum_{k_{q}=0}^{q} \Bigg((-1)^{k_{q}+k_{\ell}}\begin{pmatrix}q\\k_q\end{pmatrix}(k_{q}+k_{\ell})!X_{k_{q}+k_{\ell}}(\sigma_{i})B
  \\&
 \otimes \bigg[\left.\frac{\partial^{p+q-k_{q}}}{\partial s_{1}^{p}\partial s_{2}^{q-k_{q}}}Z_{21}(s_{1},s_{2})\right|_{\Shortstack{$\scriptstyle s_{1}=\sigma_{i}\enskip$ $\scriptstyle s_{2}=2\sigma_{i}$ }}  +\left. \frac{\partial^{p+q-k_{q}}}{\partial s_{1}^{p}\partial s_{2}^{q-k_{q}}}Z_{22}(s_{1},s_{2})\right|_{\Shortstack{$\scriptstyle s_{1}=\sigma_{i}\enskip$ $\scriptstyle s_{2}=2\sigma_{i}$ }}\bigg] \Bigg) \\
= C& \left. \frac{\partial^{p+q+\ell}}{\partial s_{1}^{p}\partial s_{2}^{q}\partial s_{3}^{\ell}}Z_{32}(s_{1},s_{2},s_{3})\right|_{\Shortstack{$\scriptstyle s_{1}=\sigma_{i}\enskip$ $\scriptstyle s_{2}=2\sigma_{i}$ $\scriptstyle s_{3}=3\sigma_{i}$ }}.
\end{align*}
Similarly for the third component $\hat{Z}_{33}(s_{1},s_{2},s_{3})$, we have
\begin{align*}
&\left.\hat{C}\frac{\partial^{p+q+\ell}}{\partial s_{1}^{p}\partial s_{2}^{q}\partial s_{3}^{z}}\hat{Z}_{33}(s_{1},s_{2},s_{3}) \right|_{\Shortstack{$\scriptstyle s_{1}=\sigma_{i}\enskip$ $\scriptstyle s_{2}=2\sigma_{i}$  $\scriptstyle s_{3}=3\sigma_{i}$}}\\
=& \hat{C}\sum_{k_{\ell}=0}^{\ell}(-1)^{\ell-k_{\ell}}\begin{pmatrix}\ell\\k_\ell\end{pmatrix}(\ell-k_{\ell})!\hat{X}_{\ell-k_{\ell}}(3\sigma_{i})\hat{H} \sum_{k_{p}=0}^{p} \Bigg(
  \bigg[\frac{\partial^{k_{p}+q+k_{\ell}}}{\partial s_{1}^{k_{p}}\partial s_{2}^{q}\partial s_{3}^{k_{\ell}}}\hat{Z}_{21}(s_{2}-s_{1},s_{3}-s_{1}) 
\\& +\left. \frac{\partial^{k_{p}+q+k_{\ell}}}{\partial s_{1}^{k_{p}}\partial s_{2}^{q}\partial s_{3}^{k_{\ell}}}\hat{Z}_{22}(s_{2}-s_{1},s_{3}-s_{1}) \bigg] \right|_{\Shortstack{ $\scriptstyle s_{1}=\sigma_{i}\enskip$ $\scriptstyle s_{2}=2\sigma_{i}$ $\scriptstyle s_{3}=3\sigma_{i}$}}
 \otimes (-1)^{p-k_{p}}\begin{pmatrix}p\\k_p\end{pmatrix}(p-k_{p})!\hat{X}_{p-k_{p}}(\sigma_{i})\hat{B}] \Bigg) \\  
=& \hat{C}\sum_{k_{\ell}=0}^{\ell}(-1)^{\ell-k_{\ell}}\begin{pmatrix}\ell\\k_\ell\end{pmatrix}(\ell-k_{\ell})!\hat{X}_{\ell-k_{\ell}}(3\sigma_{i})W^{T}H(V_{1}\otimes V_{1})\sum_{k_{p}=0}^{p}  \Bigg(
\bigg[\frac{\partial^{k_{p}+q+k_{\ell}}}{\partial s_{1}^{k_{p}}\partial s_{2}^{q}\partial s_{3}^{k_{\ell}}}
\hat{Z}_{21}(s_{2}-s_{1},s_{3}-s_{1})\\& +\left. \frac{\partial^{k_{p}+q+k_{\ell}}}{\partial s_{1}^{k_{p}}\partial s_{2}^{q}\partial s_{3}^{k_{\ell}}}\hat{Z}_{22}(s_{2}-s_{1},s_{3}-s_{1})\bigg] \right|_{\Shortstack{$\scriptstyle s_{1}=\sigma_{i}\enskip$ $\scriptstyle s_{2}=2\sigma_{i}$ $\scriptstyle s_{3}=3\sigma_{i}$}}
 \otimes (-1)^{p-k_{p}}\begin{pmatrix}p\\k_p\end{pmatrix}(p-k_{p})!\hat{X}_{p-k_{p}}(\sigma_{i})\hat{B}] \Bigg)\\ 
=& \hat{C}\sum_{k_{\ell}=0}^{\ell}(-1)^{\ell-k_{\ell}}\begin{pmatrix}\ell\\k_\ell\end{pmatrix}(\ell-k_{\ell})!\hat{X}_{\ell-k_{\ell}}(3\sigma_{i})W^{T}H\sum_{k_{p}=0}^{p}  \Bigg(V_{1}
\bigg[ \frac{\partial^{k_{p}+q+k_{\ell}}}{\partial s_{1}^{k_{p}}\partial s_{2}^{q}\partial s_{3}^{k_{\ell}}}\hat{Z}_{21}(s_{2}-s_{1},s_{3}-s_{1})
\\&+\left. \frac{\partial^{k_{p}+q+k_{\ell}}}{\partial s_{1}^{k_{p}}\partial s_{2}^{q}\partial s_{3}^{k_{\ell}}}\hat{Z}_{22}(s_{2}-s_{1},s_{3}-s_{1})\bigg]\right|_{\Shortstack{$\scriptstyle s_{1}=\sigma_{i}\enskip$  $\scriptstyle s_{2}=2\sigma_{i}$ $\scriptstyle s_{3}=3\sigma_{i}$}}
\otimes V_{1} (-1)^{p-k_{p}}\begin{pmatrix}p\\k_p\end{pmatrix}(p-k_{p})!\hat{X}_{p-k_{p}}(\sigma_{i})\hat{B}] \Bigg)\\  
=& C\sum_{k_{\ell}=0}^{\ell}(-1)^{\ell-k_{\ell}}\begin{pmatrix}\ell\\k_\ell\end{pmatrix}(\ell-k_{\ell})!X_{\ell-k_{\ell}}(3\sigma_{i})\sum_{k_{p}=0}^{p}H  \Bigg(
\bigg[\frac{\partial^{k_{p}+q+k_{\ell}}}{\partial s_{1}^{k_{p}}\partial s_{2}^{q}\partial s_{3}^{k_{\ell}}}Z_{21}(s_{2}-s_{1},s_{3}-s_{1})
\\&
+\left. \frac{\partial^{k_{p}+q+k_{\ell}}}{\partial s_{1}^{k_{p}}\partial s_{2}^{q}\partial s_{3}^{k_{\ell}}}Z_{22}(s_{2}-s_{1},s_{3}-s_{1})\bigg]\right|_{\Shortstack{$\scriptstyle s_{1}=\sigma_{i}\enskip$ $\scriptstyle s_{2}=2\sigma_{i}$  $\scriptstyle s_{3}=3\sigma_{i}$}}
\otimes (-1)^{p-k_{p}}\begin{pmatrix}p\\k_p\end{pmatrix}(p-k_{p})!X_{p-k_{p}}(\sigma_{i})B] \Bigg)\\
=&\left.C\frac{\partial^{p+q+\ell}}{\partial s_{1}^{p}\partial s_{2}^{q}\partial s_{3}^{\ell}}Z_{33}(s_{1},s_{2},s_{3})\right|_{\Shortstack{$\scriptstyle s_{1}=\sigma_{i}\enskip$  $\scriptstyle s_{2}=2\sigma_{i}$ $\scriptstyle s_{3}=3\sigma_{i}$}}
 \end{align*}
The above results clearly show that the last equation in \eqref{regG} also holds and this completes the proof.
\end{proof}

Thus we can match higher multi-moments for the first three subsystems and the approach is summarised in Algorithm~\ref{alg:2}.

\begin{algorithm}[h]
\SetAlgoLined
\KwIn{$E, A, B, C, N, H$, $\sigma_{i} \in \mathbb{C}$ for $i=1,\ldots,m$, $P, Q, L\in\mathbb{N}$}
\KwOut{$\hat{E}, \hat{A}, \hat{B}, \hat{C}, \hat{N}, \hat{H}$}
Compute $V$ and $W$ as in Algorithm~\ref{alg:1} with $k_{v}=0:P+\max(Q,L)$ and $k_{w}=0:\max(P,Q)$. \\
$V_{1}=V; W_{1}=W; Q_{max}=\max(P+L,Q); P_{max}=\max(P,Q+L) $;\\
\For{$ 	Q_{2}=0:Q_{max}$}{
\For{$P_{2}=0:P_{max}$}{
\begin{equation*}
\starspan(V_{1})=\starspan_{i=1,\ldots,m}([V_{1}, ~ [V_{N}(\sigma_{i},P_{2},Q_{2})+V_{H}(\sigma_{i},P_{2},Q_{2})]])
\end{equation*}}}
\For{$\ell=0:L$}{
\begin{equation*}
\starspan(W_{1})=\starspan_{i=1,...,m}([W_{1}, ~X_{\ell}(3\sigma_{i})^{T}C^T])
\end{equation*}}
 $U=orth([V_1, W_1])$; \\
Compute the reduced-order matrices
\begin{align*}
\hat{E}&=U^{T}EU,& \hat{A}&=U^{T}AU,& \hat{B}&=U^{T}B,&\\
\hat{C}&=U^{T}C,& \hat{N}&=U^{T}NU,& \hat{H}&=U^{T}H(U\otimes U).
\end{align*}\\
{\bf return} $\hat{E}, \hat{A}, \hat{B}, \hat{C}, \hat{N}, \hat{H}$   
\caption{Multi-moment matching for three regular transfer functions}
\label{alg:2}
\end{algorithm}

\section{Numerical Examples}
\label{sec:comparison}
We consider three benchmark examples: the non-linear RC circuit \cite{morRewW03}, the 1D Burgers' equation \cite{morKunV08} and the FitzHugh-Nagumo system \cite{morChaS10} for comparing our results. Three recent nonlinear reduction techniques are used for testing, that includes implicit symmetric moment matching (as discussed in Theorems \ref{thbenner}), generalized implicit symmetric moment matching (discussed in Theorem \ref{thtnf}), and implicit regular moment matching. The results of the first two methods will be represented by (imm-s\protect\greenline) and (igmm-s\protect\blueline). The proposed method, which involves generalized implicit moment matching for regular subsystems is represented by (igmm-r2 \protect\magentaline) and (igmm-r3 \protect\redline) for the first two and three subsystems, respectively.\par
\subsection{Non-linear RC Circuit} It is a well-used example for model order reduction of nonlinear systems, where a non-linear RC circuit as shown in Figure \ref{figrc1} is considered for model reduction.\par
\begin{figure}
\includegraphics[width=1\linewidth, height=0.295123\linewidth]{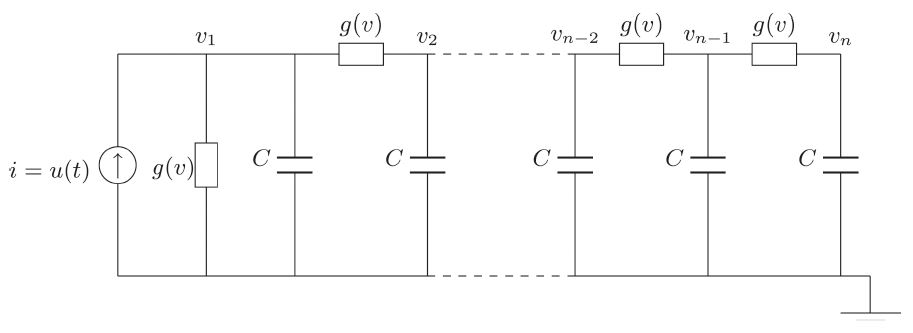}
\caption{RC circuit diagram}
\label{figrc1}
\centering
\end{figure}
The non-linearity is due to the diode $I$-$V$ characteristics, expressed as $g(v)=e^{40v}-1$, where $v$ is the node voltage. The voltage $v_{1}(t)$ at node 1 is taken as the output and the current $i$ as the input of the system. After applying Kirchoff's current law at each of the $N$ nodes, with the assumption of a normalised capacitance $C=1$, we obtain
\begin{align*}
\dot{v}(t)=f(v(t))+bu(t), && y(t)=cv(t),
\end{align*}
in which $b=c^{T}$ is the first column of the $N\times N$ identity matrix. Transforming this non-linear model to an equivalent quadratic-bilinear descriptor system increases the size to $n=2N$. For our results, we set $N=1250$, so that the order of the quadratic-bilinear system is $n = 2500$.\par

The order of the system is reduced to $r\approx 23$ with each of the four techniques, (imm-s2), (igmm-s2), (igmm-r2) and (igmm-r3). The 2 and 3 at the end denote the first two and three transfer functions are matched by each technique, respectively. The parameters used by the model reduction techniques are reported in Table \ref{tabelrc} where $r$ is the reduced order, $m$ is the number of interpolation points and $\sigma$ represent the interpolation points. The order $\mathcal{O} (p,q,\ell)$ shows the number of higher-order moments matched with respect to $s_{1}$, $s_{2}$ and $s_{3}$. The interpolation points $\sigma$ are similar to those used in \cite{morYanJ18}, so that we can reproduce their results and perform comparison with the proposed technique. The output of the original and the reduced system for an exponential input $e^{-t}$ along with their relative errors $(\frac{|y(t)-y_r(t)|)}{|y(t)|})$ for each reduction technique are shown in Figure \ref{rc2_exp}. The maximum errors $e_{max}$ for each reduced techniques are shown in Table \ref{tableRCm}.
\begin{table}
\centering
\begin{tabular}{|c|c|c|c|c|}
\hline
MOR technique                    & $r$ & $m$ & $\sigma$          & $\mathcal{O} (p,q,\ell) $ \\ \hline
(imm-s2)				    & 23  & 5          & {[}0.01,1,10,100,1000,10000{]} & $(1,1,-)$         \\ \hline
(igmm-s2)    				    & 24  & 2          & {[}0.1,10,1000{]} & $(2,2,-)$         \\ \hline
(igmm-r2)                      & 22  & 2          & {[}0.1,10,1000{]} & $(3,2,-)$         \\ \hline
(igmm-r3)                      & 23  & 1          & {[}0.1,10{]} & $(1,1,2)$          \\ \hline
\end{tabular}
\caption{Parameters for Non-linear RC-circuit }
\label{tabelrc}
\end{table}
\begin{table}
\centering
\begin{tabular}{|c|c|c|c|}
\hline
MOR technique 	& $r$ & $e_{max}$ for $u(t)=e^{-t}$	&$e_{max}$ for $u(t)=cos(2\pi (t/10)+1)/2$\\
\hline
(imm-s2) 	& 23 & $0.272 \times 10^{-6}$	&$0.414 \times 10^{-6^{•}}$\\
\hline
(igmm-s2) 	& 24 & $0.1715 \times 10^{-5}$	&$0.3899 \times 10^{-5}$\\
\hline
(igmm-r2)	& 22 & $0.982 \times 10^{-6}$	&$0.2274 \times 10^{-5}$\\
\hline
(igmm-r3) 	& 23 & $0.132 \times 10^{-6}$	&$0.133 \times 10^{-6}$\\
\hline
\end{tabular}
\caption{Maximum errors for Non-linear RC-circuit}
\label{tableRCm}
\end{table}
\begin{figure}
\begin{center}
\subfloat[Output Response]{\label{rc2_a}\includegraphics[width=0.45\linewidth, height=0.35\linewidth]{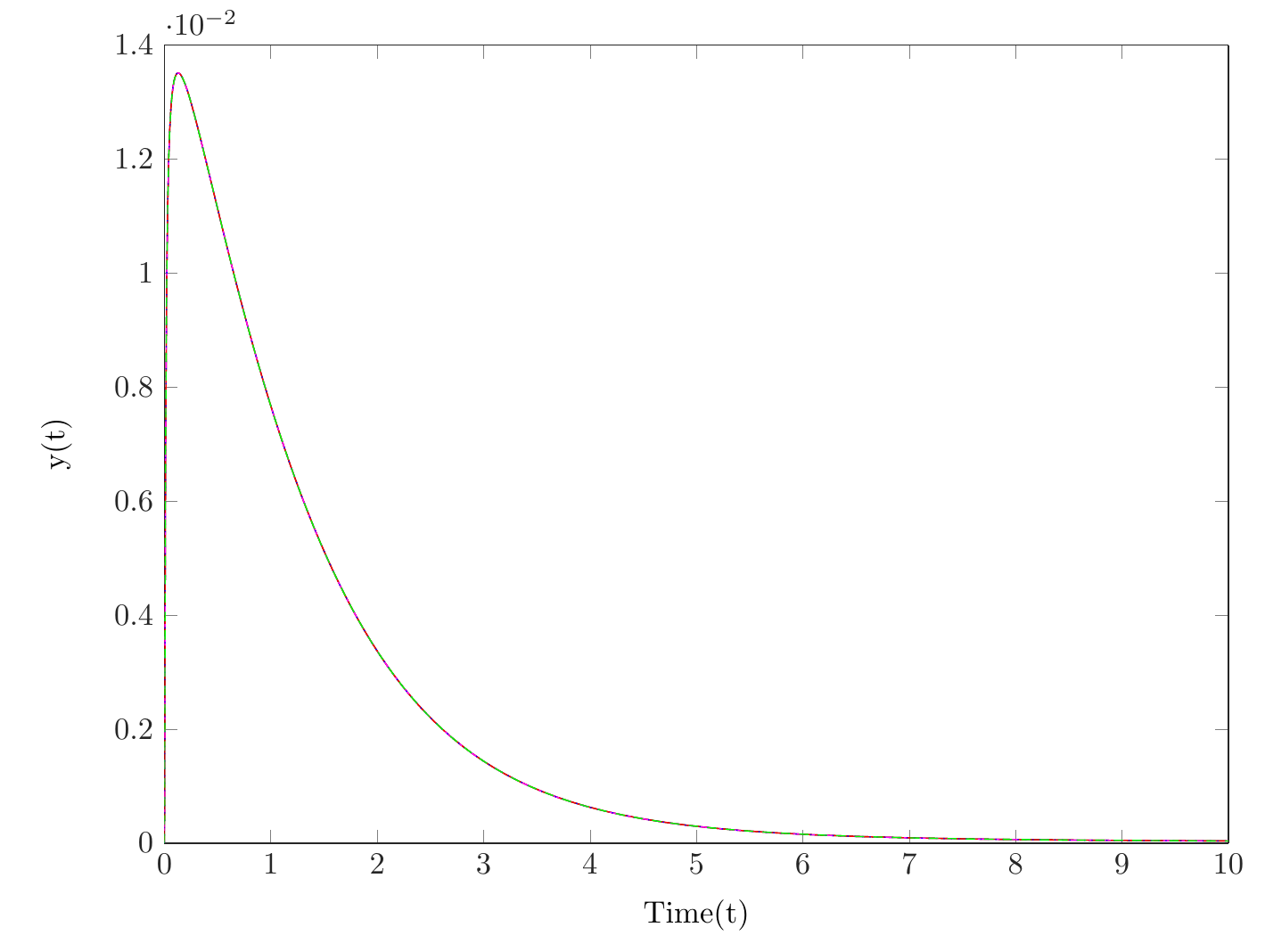}}
\subfloat[Relative Error]{\label{rc2_b}\includegraphics[width=0.45\linewidth, height=0.35\linewidth]{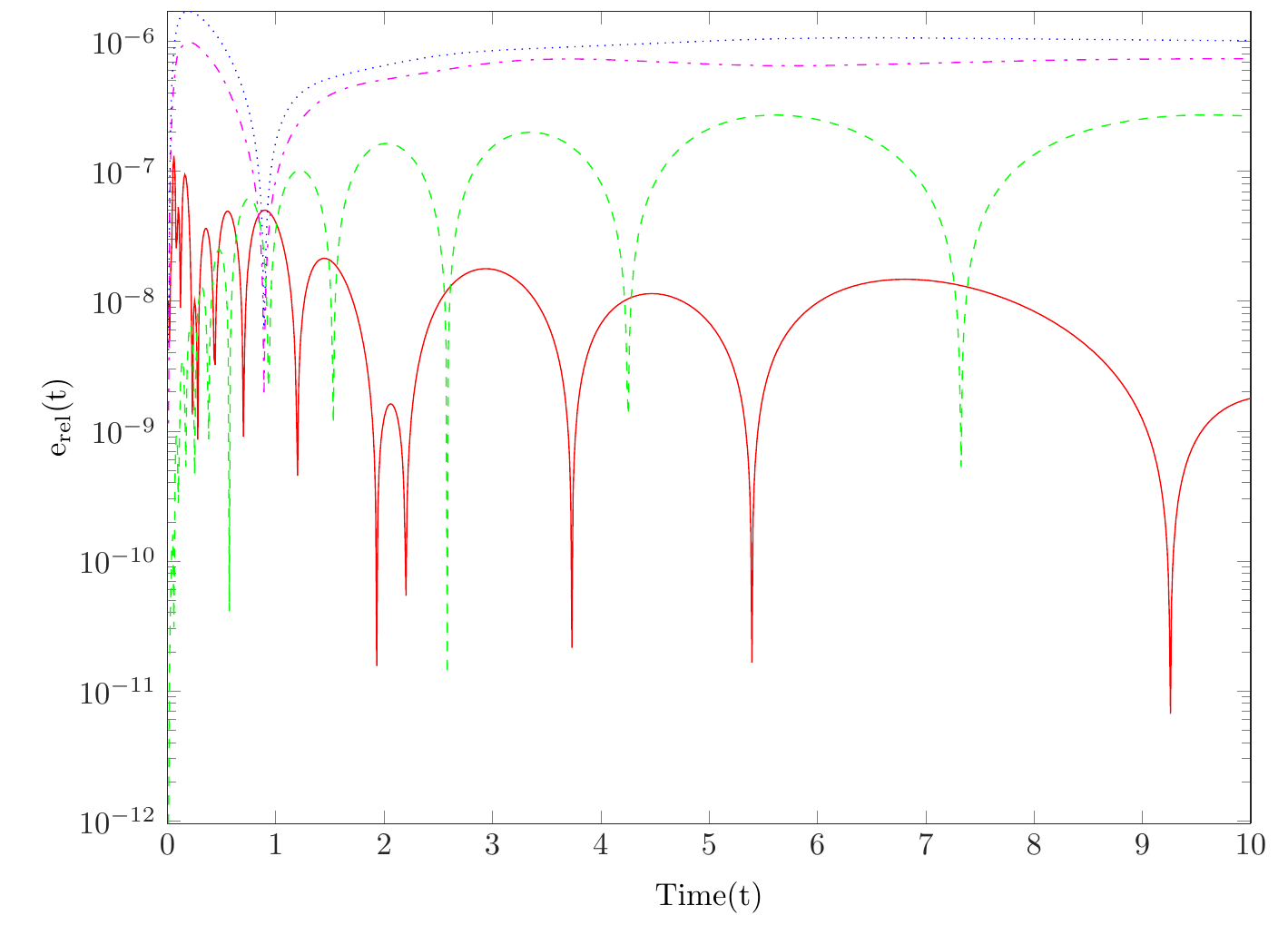}}
\caption{Non-linear RC-circuit with $u(t)=e^{-t}$; Reduced Order Model (ROM) via imm-s (\protect\greenline), igmm-s (\protect\blueline), igmm-r2 (\protect\magentaline) and  igmm-r3 (\protect\redline).}
\label{rc2_exp}
\end{center}\end{figure}
Overall the results are similar for those techniques that match the first two transfer functions; imm-s2, igmm-s2 and igmm-r2. The higher-order moment matching techniques; however, do not use nonlinear matrices $N$ and $H$ in the construction of the projection matrices $V$ and $W$. Also, fewer interpolation points are involved, which makes the choice of interpolation points relatively easy.\par
If we change the input to $u(t)=cos(2\pi (t/10)+1)/2$, the results are shown in Figure \ref{rcsin}. It is clear that the change in input shows similar results without constructing the reduced model, showing that the reduced model is input independent. \par
\begin{figure}
\begin{center}
\subfloat[Output]{\label{rc2_c}\includegraphics[width=0.45\linewidth, height=0.35\linewidth]{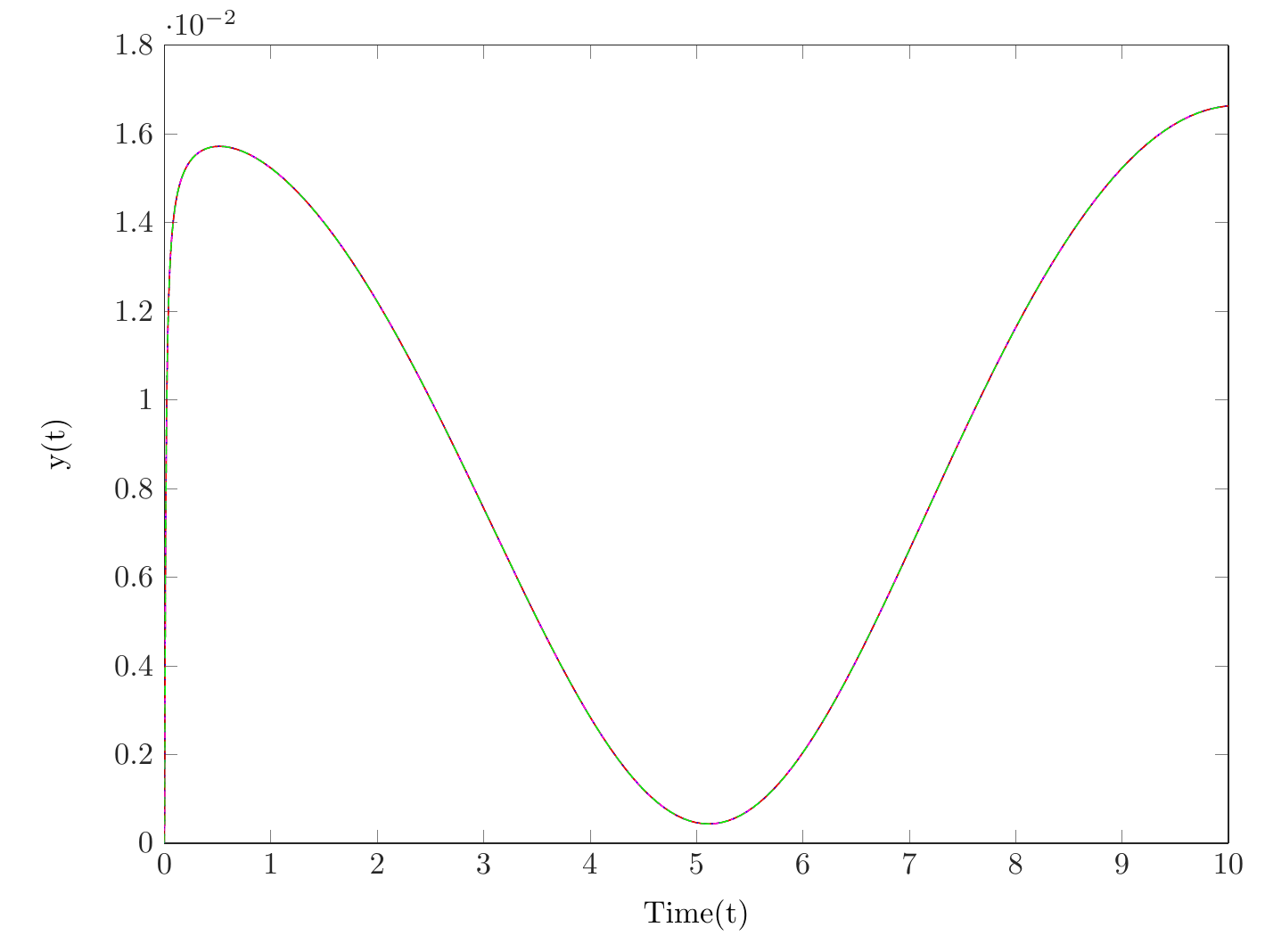}}
\subfloat[Relative Error]{\label{rc2_d}\includegraphics[width=0.45\linewidth, height=0.35\linewidth]{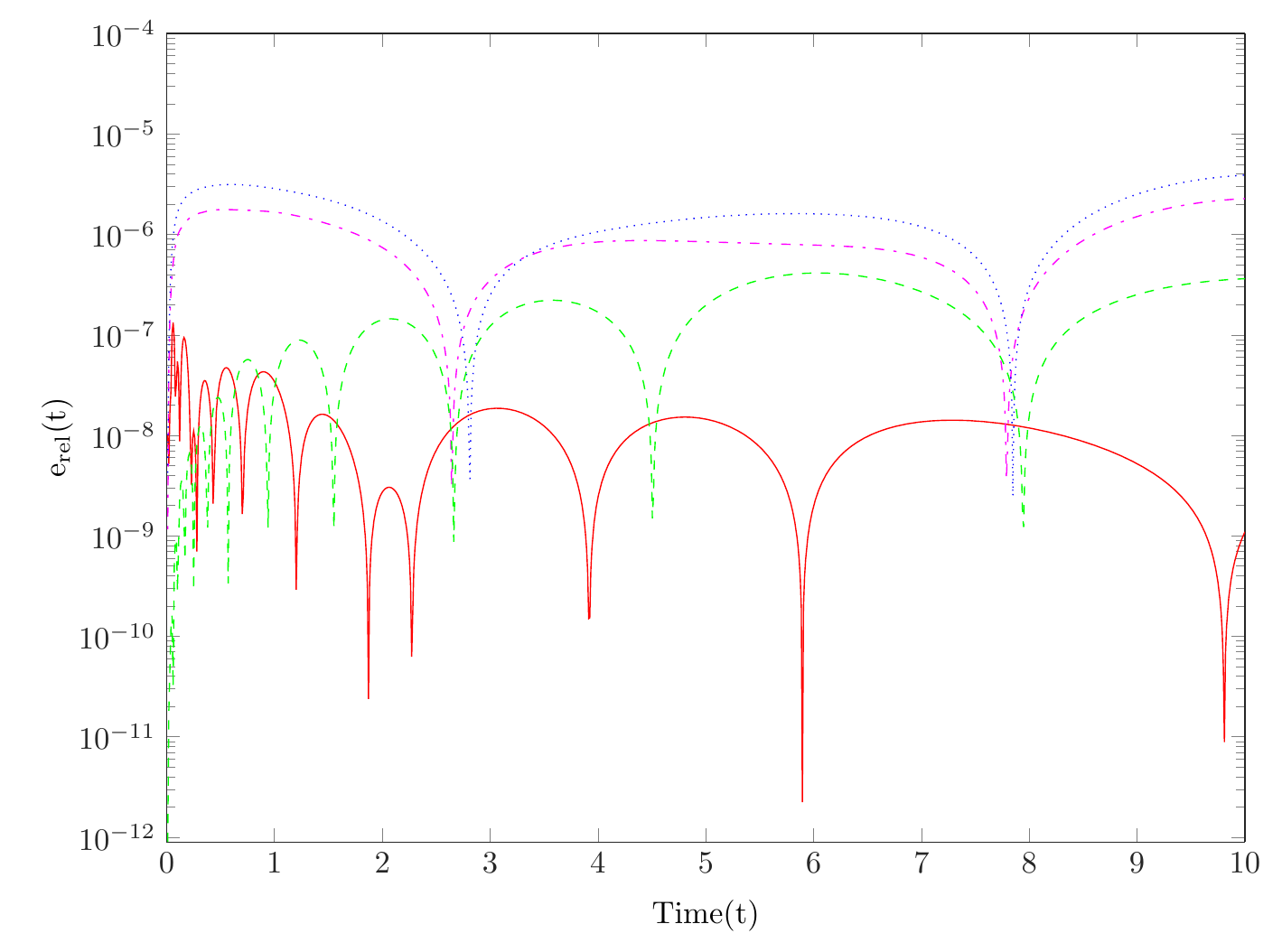}}
\caption{ Non-linear RC-circuit with $u(t)=cos(2\pi (t/10)+1)/2$; ROM via imm-s (\protect\greenline), igmm-s (\protect\blueline), igmm-r2 (\protect\magentaline) and  igmm-r3 (\protect\redline). }\label{rcsin}
\end{center}\end{figure}

The results in Figures \ref{rc2_exp} and \ref{rcsin} clearly show that the higher transfer function projection technique (igmm-r3), using the regular form, gives better reduced order models comparatively. Also, it is important to note that (igmm-r3) is not only ensuring higher-order moment matching but also uses the matrices from the nonlinear terms  in the construction of the matrices $V$ and $W$. This improves the quality of the reduced model. 
\subsection{Burgers' Equation}
\begin{sloppypar}
As a second example, we consider the boundary control problem for 1D Burgers' equation on the domain ${\Omega = (0,1)\times (0,T)}$, that results in the following set of equations \cite{morKunV08}:
\end{sloppypar}
\begin{align}
\begin{split}
v_{t}+v\cdot v_{x}                   & = \nu \cdot v_{xx} \ \ in \ (0,1)\times (0,T), \\
\alpha v(0,\cdot )+\beta v(0,\cdot ) & = u(t)           \ \ \ \ t \in (0,T), \\
v_{x}(1,t)                           & = 0,             \ \ \ \ \   t \in   (0,T), \\
v(x,0)                               & = v_{0}(x)       \ \   x \in  (0,T).
\end{split}
\end{align}
Here $\nu$  is the viscosity and  $v(0,x)$ is the initial condition of the system. The above PDE can be converted to the desired quadratic-bilinear system using standard discretisation. We select $\nu=0.05$ and $n=1000$ and then compute a reduced model of order $r$ using different techniques.\par
The size of the required reduced order model $r$ is fixed to $23$ by carefully choosing the number of higher multi-moments being matched. Using the reduction techniques discussed above, the response and the absolute error are shown in Figures \ref{b2_exp} and \ref{b2} for $u(t)=e^{-t}$ and $u(t)=cos(2\pi (t/10)+1)/2$, respectively. The maximum errors $e_{max}$ are shown in Table \ref{tablebm} and the parameters used in each technique are shown in Table \ref{tabelb}. \par
\begin{table}
\centering
\begin{tabular}{|c|c|c|c|}
\hline
MOR technique 	& $r$ &$e_{max}$ for $u(t)=e^{-t}$	&$e_{max}$ for $u(t)=cos(2\pi (t/10)+1)/2$\\
\hline
(imm-s2) & 23 	& $0.0083$	&$0.0377$\\
\hline
(igmm-s2) & 24 	& $0.0027$	&$0.0254$\\
\hline
(igmm-r2)	& 22  & $0.0019$	&$0.0227$\\
\hline
(igmm-r3) 	& 23 & $0.0007$	&$0.0063$\\
\hline
\end{tabular}
\caption{Maximum errors for Burgers' Equation }
\label{tablebm}
\end{table}
\begin{table}
\centering
\begin{tabular}{|c|c|c|c|c|}
\hline
MOR technique                    & $r$ & $m$ & $\sigma$          & $\mathcal{O} (p,q,\ell) $ \\ \hline
(imm-s2)				    & 23  & 5          & {[}0.01,1,10,100,1000,10000{]} & $(1,1,-)$         \\ \hline
(igmm-s2)    				    & 24  & 2          & {[}0.1,10,1000{]} & $(2,2,-)$         \\ \hline
(igmm-r2)                      & 22  & 2          & {[}0.1,10,1000{]} & $(3,2,-)$         \\ \hline
(igmm-r3)                      & 23  & 1          & {[}0.1,10{]} & $(1,1,2)$          \\ \hline
\end{tabular}
\caption{Parameters for Burgers' Equation }\label{tabelb}
\end{table} 

\begin{figure}
\begin{center}
\subfloat[Output]{\label{b2_a}\includegraphics[width=0.45\linewidth, height=0.35\linewidth]{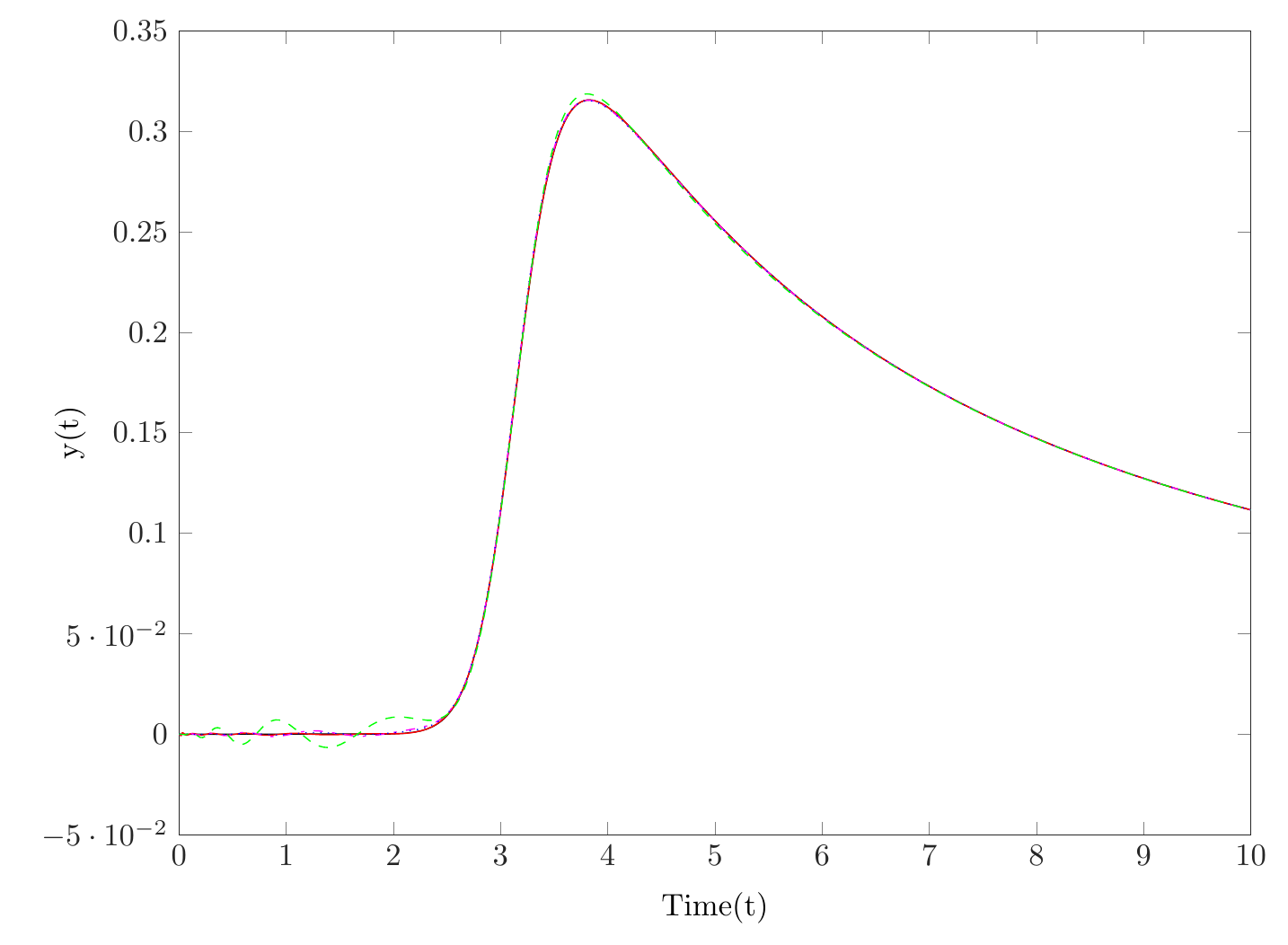}}
\subfloat[Absolute Error]{\label{b2_b}\includegraphics[width=0.45\linewidth, height=0.35\linewidth]{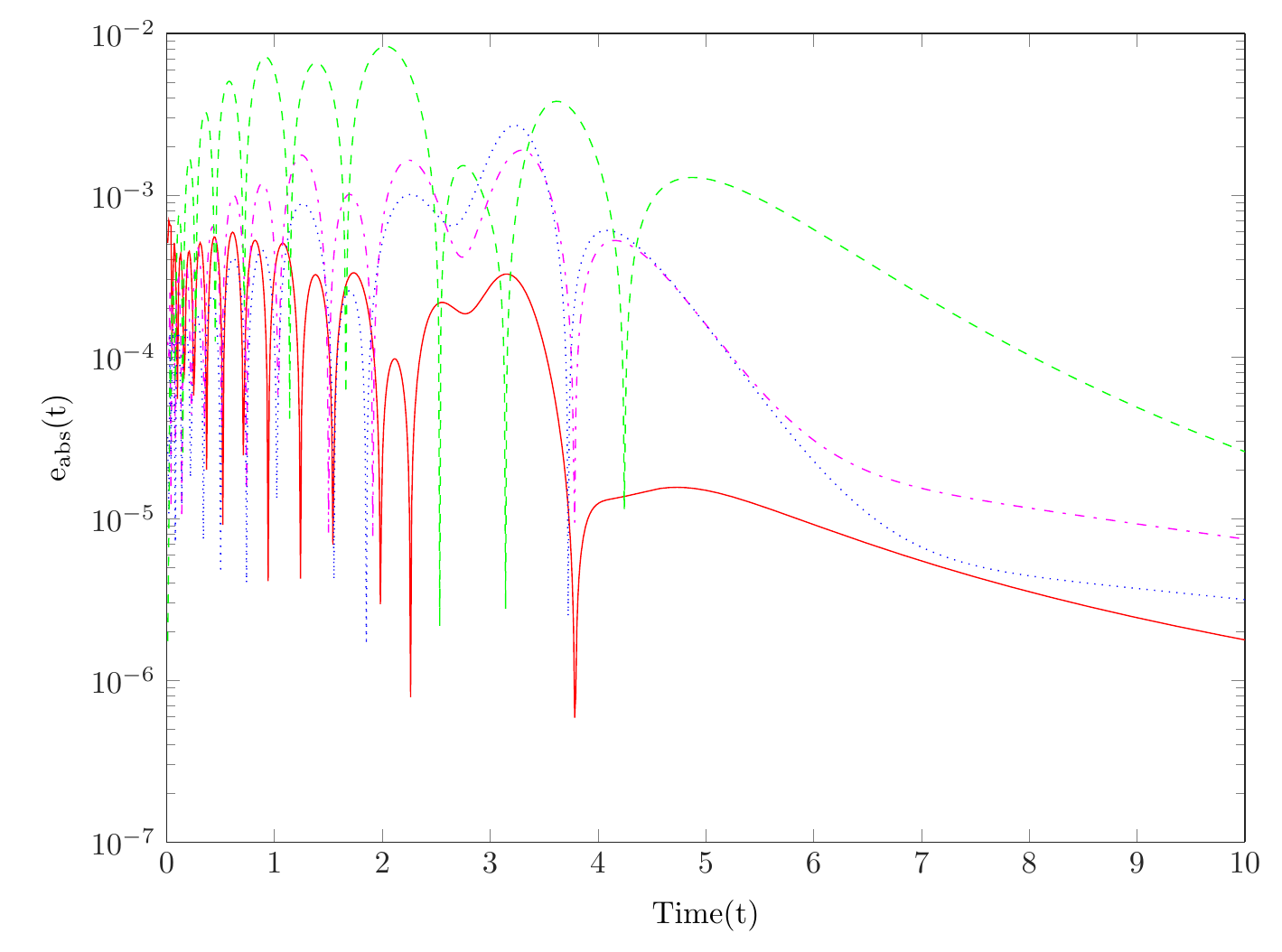}}
\caption{Burgers' Equation with $u(t)=e^{-t}$; ROM via imm-s (\protect\greenline), igmm-s (\protect\blueline), igmm-r2 (\protect\magentaline) and  igmm-r3 (\protect\redline).}\label{b2_exp}
\end{center}\end{figure}
\begin{figure}
\begin{center}
\subfloat[Output]{\label{b2_c}\includegraphics[width=0.45\linewidth, height=0.35\linewidth]{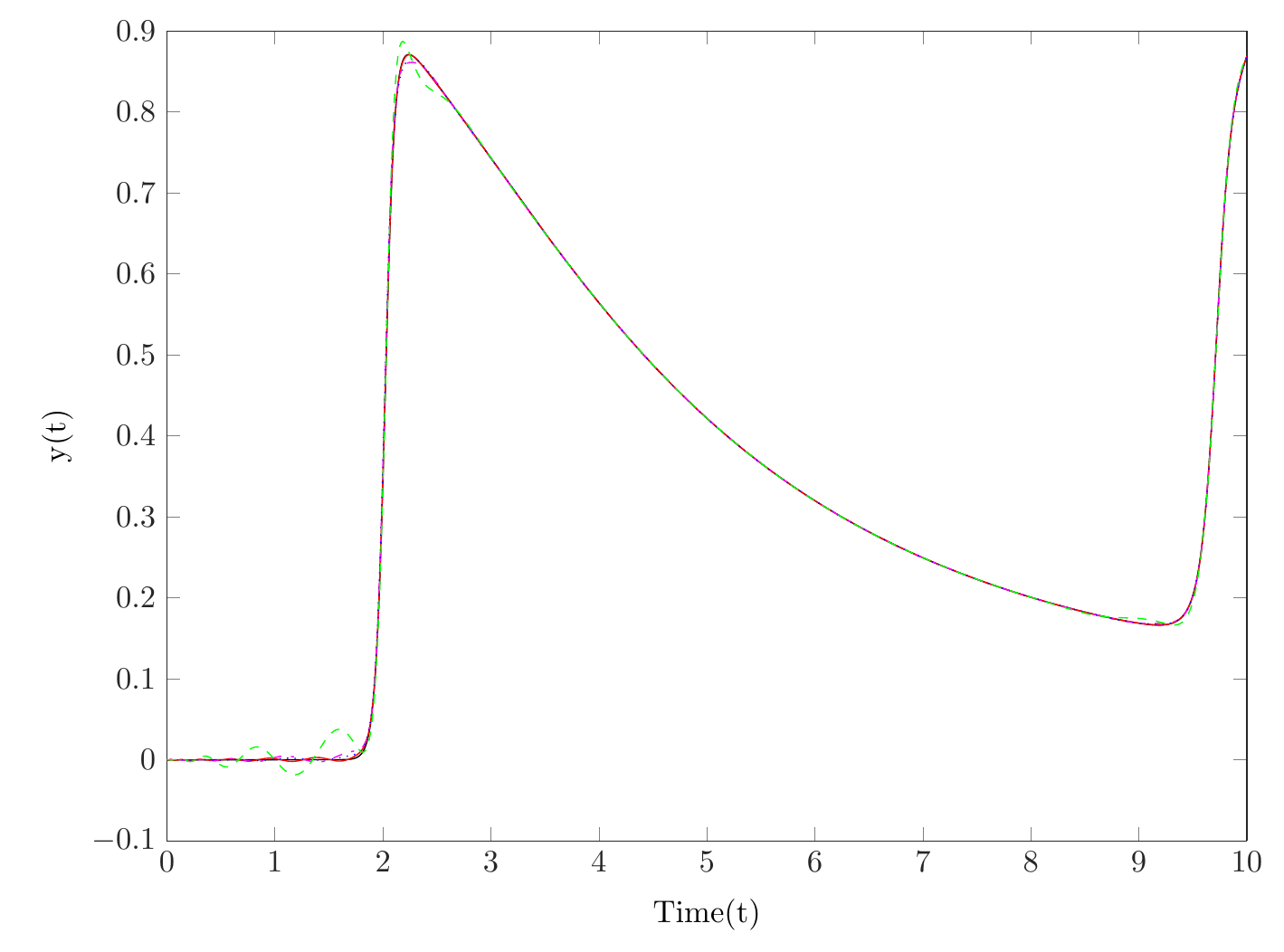}}
\subfloat[Absolute Error]{\label{b2_d}\includegraphics[width=0.45\linewidth, height=0.35\linewidth]{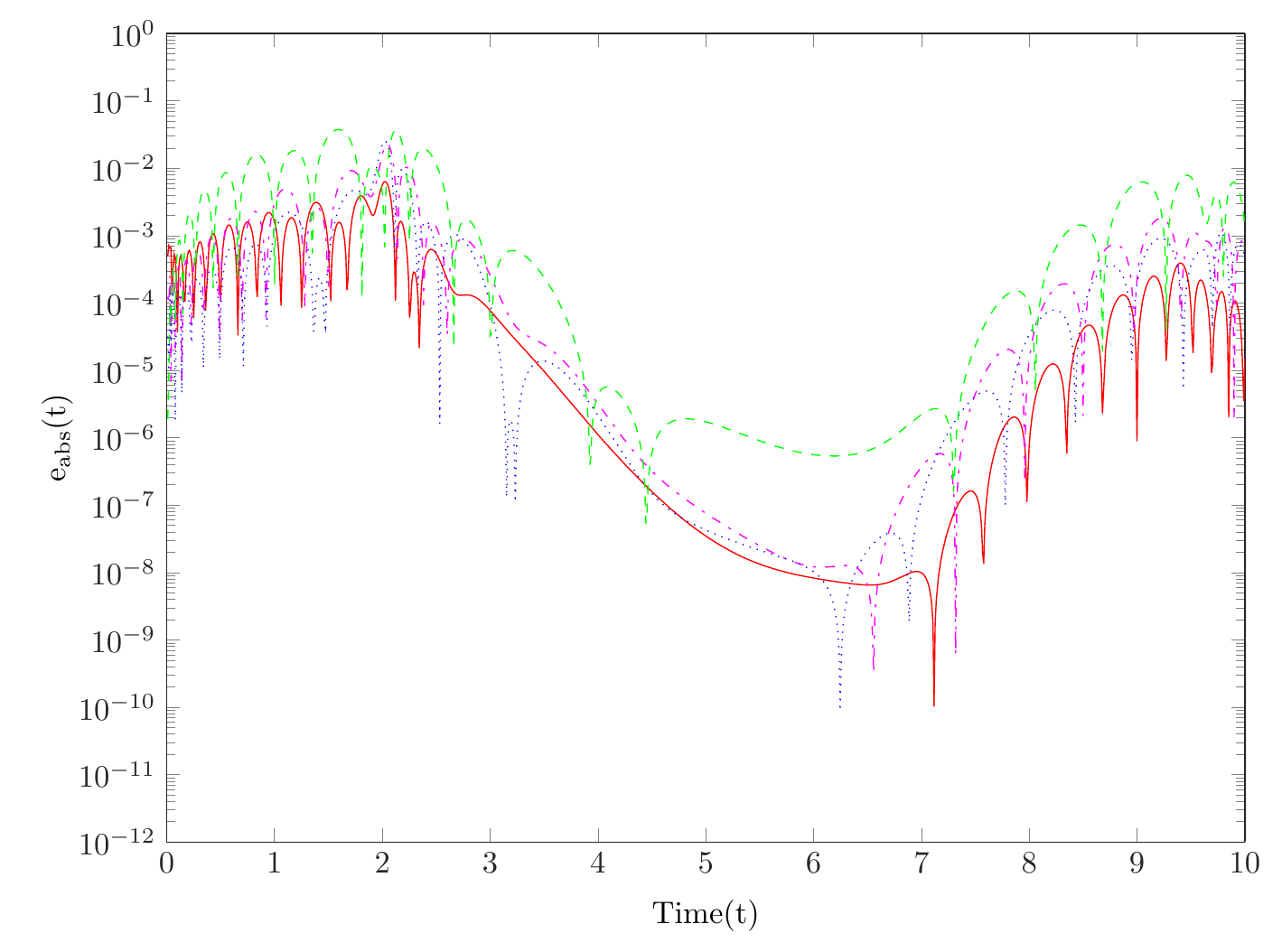}}
\caption{Burgers' Equation with $u(t)= cos(2\pi (t/10)+1)/2$; ROM via imm-s (\protect\greenline), igmm-s (\protect\blueline), igmm-r2 (\protect\magentaline) and  igmm-r3 (\protect\redline). }\label{b2}
\end{center}
\end{figure}
%
Although the results here are comparable, the proposed method (igmm-r3) is slightly better than the existing techniques. It is observed that increasing the value of $\mathcal{O} (p,q,\ell)$, reduces the absolute error to a much smaller value as more multi-moments are matched. 
\subsection{FitzHugh-Nagumo System}
Our last example is the FitzHugh-Nagumo system where the activation and deactivation dynamics of a spiking neuron are modeled by coupled non-linear PDEs
\begin{align}
\label{fhn}
\begin{split}
\epsilon v_{t}(x,t)&=\epsilon^{2}v_{xx}(x,t)+f(v(x,t))-w(x,t)+g,\\
w_{t}(x,t)&=hv(x,t)-\gamma w(x,t)+g,
\end{split}
\end{align}
with $f(v)=v(v-0.1)(1-v)$ and boundary conditions
\begin{align*}
v(x,0)&=0, &w(x,0)&=0, &x\in [0,1],\\
v_{x}(0,t)&=-i_{o}(t), &v_{x}(1,t)&=0, &t\geq 0, 
\end{align*}
where $\epsilon =0.015$, $h=0.5$, $\gamma=0.05$ and $i_{o}(t)=5\times 10^{4}t^{3}e^{-15t}$. A system of ODEs with cubic nonlinearities are obtained by applying the standard finite difference  method. With the introduction of a new variable $z_{i}=v_{i}^{2}$; $\dot{z}_{i}=2v_{i}\dot{v}_{i}$, the cubic terms can be represented in quadratic form. This gives a quadratic-bilinear system which will involve three dynamical variables: $v_{i}$, $w_{i}$ and $z_{i}$. For $\bar{n}$ discretization points, we get a quadratic-bilinear system of dimension $n=3\bar{n}$.  Note that to incorporate the effects of the variable $g$ and the initial stimulus, the QBDAE system needs to have two inputs and two outputs.  
\par In our experiments, we set $\bar{n}=750$, so that $n=2250$, for which we used linear IRKA from \cite{morGugAB08} to identify interpolation points $\sigma_{i}$. The reduced model was getting unstable by using (imm-s) for order greater than $16$, so we selected the first four points to obtain a reduce model of order $16$ that reproduced the results in \cite{Ahmad2019}. The proposed method (igmm-r3) is then used with the first two interpolation points from IRKA to obtain a reduced model of order $24$. The results are shown in Figure \ref{f_a} where it is clear that the maximum error of (imm-s2) is significantly larger than (igmm-r3). While the sizes of the two reduced models are different, the comparison clearly shows the usefulness and accuracy of (igmm-r3). A 3-D plot for the limit cycle behavior of the original and the reduced systems is presented in Figure \ref{f_b}. Since the FitzHugh-Nagumo system has two inputs and two outputs, both (imm-s2) and (igmm-r3) are using tangential directions, that correspond to the selected interpolation points from linear IRKA.

\begin{figure}
\centering
\subfloat[Transient Response]{\label{f_a}\includegraphics[width=0.45\linewidth, height=0.35\linewidth]{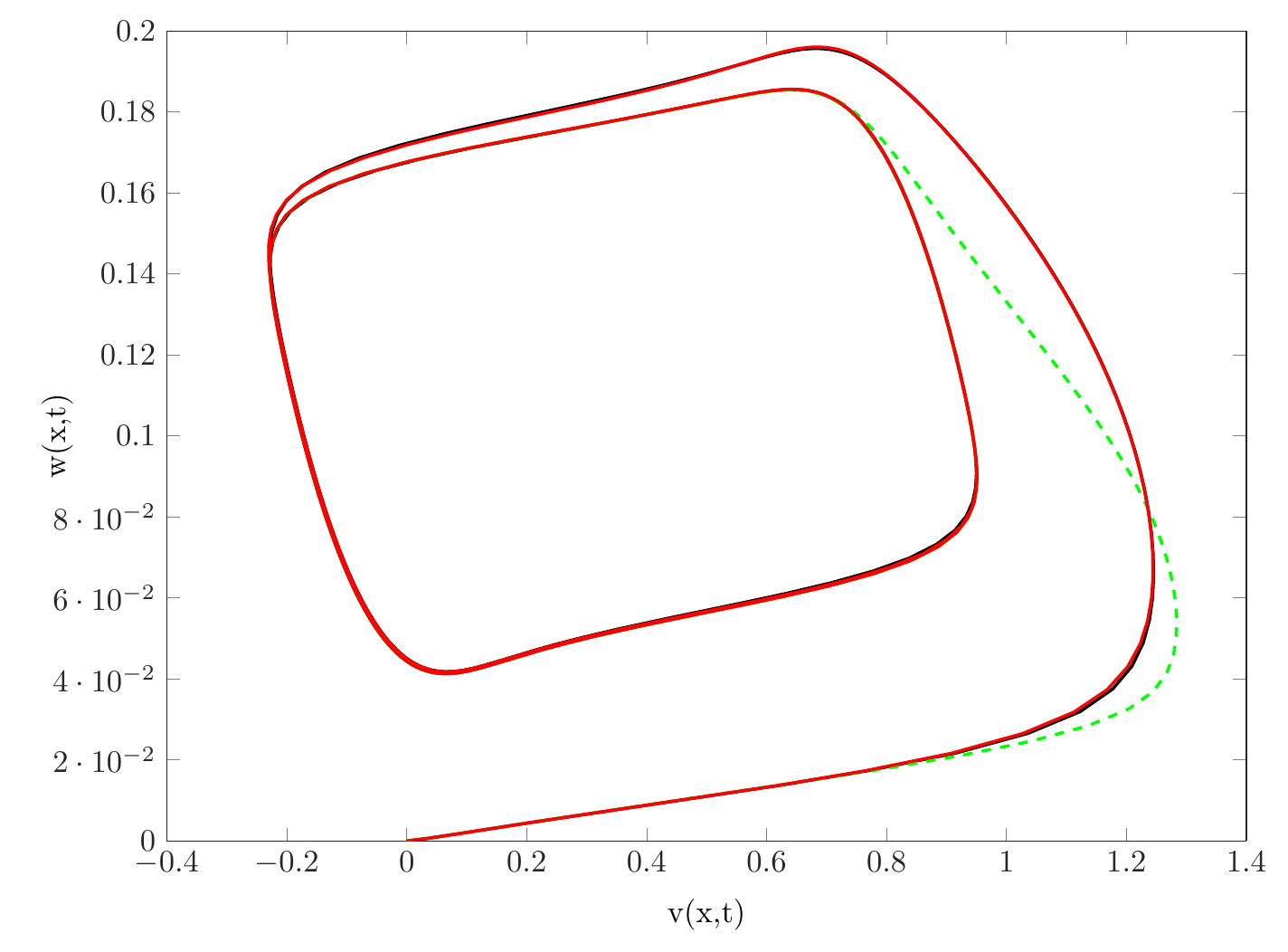}}
\subfloat[3-D View of Transient Response]{\label{f_b}\includegraphics[width=0.45\linewidth, height=0.35\linewidth]{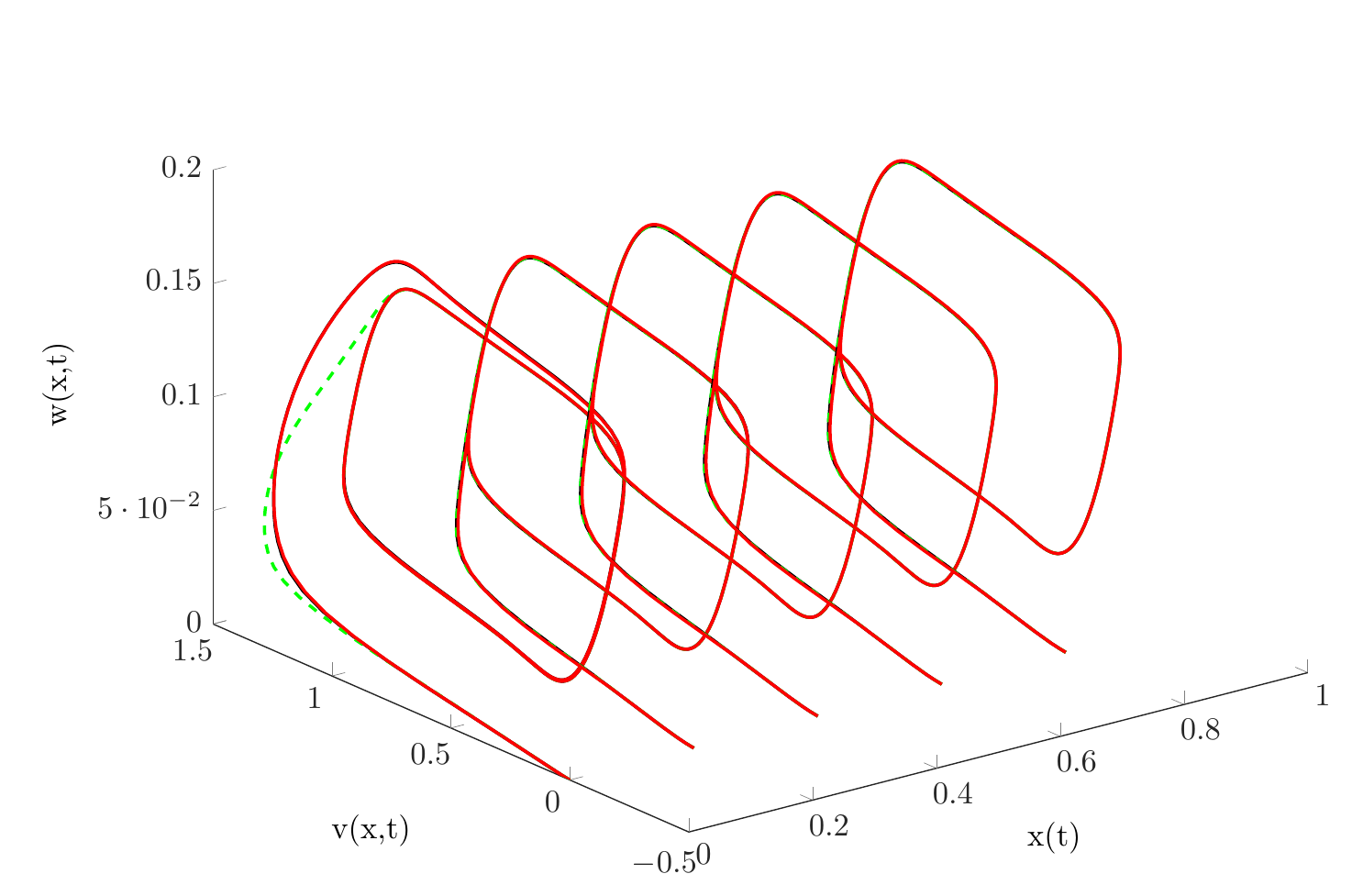}}
\caption{FitzHugh-Nagumo system with $u(t)=5\times 10^{4}t^{3}exp(-15t)$; ROM via (imm-s \protect\greenline) and  (igmm-r3 \protect\redline).}
\centering
\end{figure}
\section{Conclusion}
We showed the use of regular multivariate transfer functions for model reduction of quadratic-bilinear systems through generalized multi-moment matching by rational projection. Existing techniques are restricting the concept of generalized moment matching to the first two symmetric transfer functions. We have extended it to the first three transfer functions and in the regular form. The approximation quality of the reduced models are compared and it is observed that matching the third transfer function improves the results. An important future work could be the choice of the interpolation points and tangent directions for the generalized multi-moment matching problem. 
\bibliographystyle{cas-model2-names}
\bibliography{r}

\end{document}